%% file: main.tex
\documentclass[acmsmall]{acmart}
\AtBeginDocument{%
  }

\setcopyright{rightsretained}
\acmJournal{PACMMOD}
\acmYear{2026} \acmVolume{4} \acmNumber{1 (SIGMOD)} \acmArticle{55} \acmMonth{2} \acmPrice{}\acmDOI{10.1145/3786669}

\acmJournal{PACMMOD}
\acmVolume{4}
\acmNumber{1}
\acmArticle{55}
\acmMonth{2}

\usepackage{float}
\usepackage{graphicx}
\usepackage[caption=false]{subfig}
\graphicspath{{samples/sigmod2026/figs/}}

\usepackage[linesnumbered,vlined,boxed,ruled]{algorithm2e}
\SetVlineSkip{2pt}

\SetCommentSty{myAlgComment}
\SetKwProg{Fn}{Function}{:}{}

\usepackage{booktabs} 
\usepackage{tabularray}

\usepackage{enumitem}

\usepackage[capitalise]{cleveref}
\crefname{algorithm}{Alg.}{Algs.}

\def\Snospace~{\S{}}

\crefformat{section}{\S#2#1#3}

\usepackage{mathdots} 
\usepackage{tikz}
\usetikzlibrary{arrows.meta, decorations.pathreplacing, calligraphy, positioning, backgrounds, fit, calc}

\newcommand{\E}{\mathbb{E}}
\newcommand{\indr}[1]{\mathbf{1}_{#1}}

\newcommand{\Tree}{\mathbb{T}}

\newcommand{\CM}{\mathcal{M}}
\newcommand{\CD}{\mathcal{D}}

\newcommand{\CO}{\mathcal{O}}

\newcommand{\DT}{\Tree^d}
\newcommand{\hDT}{\hat{\Tree}^d}
\newcommand{\AT}{\Tree^a}
\newcommand{\hAT}{\hat{\Tree}^a}
\newcommand{\OT}{\Tree^o}

\newcommand{\FO}{\text{FO}}

\newcommand{\dis}{\mathit{dis}}
\newcommand{\hdis}{\hat{\mathit{dis}}}
\newcommand{\err}{\mathit{err}}

\newcommand{\var}{\mathit{var}}

\newcommand{\bx}{\mathbf{x}}
\newcommand{\by}{\mathbf{y}}
\newcommand{\bc}{\mathbf{c}}

\begin{document}

\title{MTSP-LDP: A Framework for Multi-Task Streaming Data
Publication under Local Differential Privacy}


\author{CHANG LIU}
\affiliation{%
  \institution{Xi'an Jiaotong University}
  \city{SHAANXI}
  \country{China}}
\email{MollyLiu@stu.xjtu.edu.cn}
\orcid{0009-0008-2588-9949}

\author{JUNZHOU Zhao}
\affiliation{%
  \institution{Xi'an Jiaotong University}
  \city{SHAANXI}
  \country{China}}
\email{junzhou.zhao@xjtu.edu.cn}

\renewcommand{\shortauthors}{Liu and Zhao}

\begin{abstract}
The proliferation of streaming data analytics in data-driven applications raises
critical privacy concerns, as directly collecting user data may compromise
personal privacy.
Although existing $w$-event local differential privacy (LDP) mechanisms provide
formal guarantees without relying on trusted third parties, their practical
deployment is hindered by two key limitations.
First, these methods are designed primarily for publishing simple statistics at
each timestamp, making them inherently unsuitable for complex queries.
Second, they handle data at each timestamp independently, failing to capture
temporal correlations and consequently degrading the overall utility. To address these issues, we propose MTSP-LDP, a novel framework for
\textbf{M}ulti-\textbf{T}ask \textbf{S}treaming data \textbf{P}ublication under
$w$-event LDP.
MTSP-LDP adopts an \emph{Optimal Privacy Budget Allocation} algorithm to
dynamically allocate privacy budgets by analyzing temporal correlations within
each window.
It then constructs a \emph{data-adaptive private binary tree structure} to support complex queries, which is further refined by cross-timestamp grouping and smoothing operations to enhance estimation accuracy.
Furthermore, a unified \emph{Budget-Free Multi-Task Processing} mechanism is
introduced to support a variety of streaming queries without consuming additional privacy budget.
Extensive experiments on real-world datasets demonstrate that MTSP-LDP
consistently achieves high utility across various streaming tasks, significantly outperforming existing methods.
\end{abstract}



\received{July 2025}
\received[revised]{October 2025}
\received[accepted]{November 2025}

\maketitle

\input{samples/sigmod2026/introduction}

\input{samples/sigmod2026/preliminary}

\input{samples/sigmod2026/method}

\input{samples/sigmod2026/analysis}

\input{samples/sigmod2026/experiments}

\input{samples/sigmod2026/related}

\input{samples/sigmod2026/conclusion}


\bibliographystyle{ACM-Reference-Format}
\bibliography{ref}

\end{document}

%% file: samples/sigmod2026/introduction.tex
\section{Introduction} \label{sec:introduction}

With the rapid development of the Internet of Things, massive real-time data
streams have become ubiquitous in a variety of modern applications such as
traffic management and intelligent parking~\cite{Babcock2002, Babu2001}.
However, the collection and analysis of sensitive user data, such as real-time
vehicle locations in a city and user purchasing history on e-commerce websites,
pose significant privacy risks~\cite{Al2022, Ray2020, Malhotra2021}.
To protect user privacy, differential privacy (DP)~\cite{Dwork2006} has emerged
as a cornerstone, offering rigorous guarantees by injecting calibrated noise into
user data.
While centralized differential privacy (CDP)~\cite{Bolot2013} relies on a
trusted server to collect and perturb data, local differential privacy
(LDP)~\cite{Wang2017l, Wang2018p, Ren2018, Zhang2018, Duchi2013, Cormode2018,
  Wang2019l, Murakami2019, Ye2019, Qin2016, Qin2017} eliminates this requirement
by applying perturbation directly on the user side.
This advantage has led to the deployment of LDP in large-scale systems,
including Google's RAPPOR~\cite{Erlingsson2014}, Apple’s iOS
analytics~\cite{Team2017}, and Microsoft’s telemetry services~\cite{Ding2017}.

Researchers have proposed many privacy-preserving methods for streaming data,
most of which focus on user-level DP for finite streams~\cite{Fan2013a,
  Fan2013, Fan2012} or event-level DP for infinite streams~\cite{Dwork2010,
  Bolot2013}.
However, real-world applications usually run for long periods, producing
infinite streams.
Providing event-level privacy on such streams cannot satisfy the requirements for
protecting arbitrary events.
In contrast, providing user-level privacy requires adding infinite
perturbations, which ultimately reduces the utility of the data.
To address these issues, Kellaris et al.~\cite{Kellaris2014} proposed $w$-event
DP to protect any sequence of events that occur in any $w$ consecutive
timestamps (i.e., a sliding window of size $w$), achieving a balance between
privacy and utility for infinite streams.

However, existing $w$-event DP studies~\cite{Ren2020,Wang2020,Wang2019,Wang2016}
are primarily designed under the CDP setting.
As an emerging solution for infinite data streams, $w$-event LDP~\cite{Ren2022}
shows strong potential but remains underexplored.
There are two key limitations that remain understudied and hinder its practical
adoption.
First, existing mechanisms fail to achieve efficient privacy budget utilization.
While dynamic budget allocation is theoretically possible, empirical studies
reveal that these strategies often perform worse than simple baselines, such as
uniform allocation or random sampling~\cite{Schaler2024}.
Second, $w$-event LDP is severely limited in the types of analytical tasks it
supports.
While event-level and user-level DP methods have matured to support diverse
analytical tasks including counting, range queries, and event
monitoring~\cite{Chen2017, Cao2017, Chan2012}, existing $w$-event methods remain
limited to fixed-granularity statistical releases, particularly frequency
histograms~\cite{Ren2022,Li2025}.
This single-focus output makes it impossible to perform \emph{multi-granularity}
analysis on the same continuous stream.
This rigidity is a critical barrier in real-world applications such as
intelligent vehicle systems~\cite{sharif2017internet},
where operators require
insights at different \emph{spatial} resolutions.
For example, they require fine-grained insights, such as the specific vehicle
count at a single intersection to manage traffic lights.
Simultaneously, they must analyze coarse-grained data, such as the total traffic
volume along an entire arterial road to identify bottlenecks.

To overcome these limitations and design a $w$-event LDP framework that supports
concurrent multi-task processing, we identify four fundamental challenges to be
addressed:
\begin{itemize}[nosep]
\item \textbf{Spatial Imbalance in Data Distribution.}
  Real-world streaming data often exhibits significant spatial imbalance within
  individual timestamps.
  For instance, in intelligent transportation systems, traffic on arterial roads
  may far exceed that on side roads.
  Existing methods that uniformly partition the domain and apply the same noise scale to all regions can bury signals from sparse regions beneath the noise magnitude, obscuring underlying patterns and impairing downstream analytics.

\item \textbf{Temporal Variation in Data Distribution.}
  The variation in stream distributions over time poses a significant challenge for
  allocating privacy budgets that can adaptively adjust to the degree of
  variation across the sliding window.
  Existing methods, often focusing solely on local information, fail to capture these dynamics, resulting in poor utility.

\item \textbf{Privacy Budget Sharing across Multiple Tasks.}
  Streaming systems are often required to support multiple query tasks
  concurrently.
  As the number of concurrent tasks increases, the limited privacy budget must
  be divided among more tasks, causing a rapid decline in the budget allocated
  to each task and significantly degrading overall accuracy.

\item \textbf{Low-Latency Constraints in Real-Time Processing.}
  Real-time streaming applications impose strict latency requirements.
  Even if a more accurate mechanism exists, it may be impractical if it cannot
  respond within tight deadlines.
  This necessitates a trade-off among privacy protection, estimation accuracy,
  and processing latency.
\end{itemize}

To address the above challenges, we propose MTSP-LDP, a new $w$-event LDP
framework for infinite streams.
Our main contributions are summarized as follows.
\begin{itemize}[nosep]
\item {To the best of our knowledge, this paper is the first to address the
  challenge of supporting multi-task, multi-granularity analytical tasks under
  $w$-event LDP.
  Specifically, we introduce the \emph{Budget-Free Multi-Task Processing}
  mechanism, which enables concurrent support for different types of queries
  over infinite data streams.
  Notably, this mechanism operates without requiring modifications to other
  modules or the perturbation data reported by users.
  Crucially, it incurs no additional privacy cost.}

\item {Our framework's high utility is achieved through two key technical
  innovations.
  First, the \emph{Optimal Privacy Budget Allocation} mechanism adaptively
  allocates privacy budgets based on data variation across each sliding window
  of size $w$, ensuring that timestamps with greater variation receive more
  privacy budget.
  Second, the \emph{Private Adaptive Tree Publication} mechanism improves accuracy by constructing a tree tailored to the underlying data distribution, departing from traditional fixed-partition histograms. It then applies cross-timestamp grouping and smoothing to further enhance estimation accuracy. Crucially, our mechanism resolves the prohibitive level-by-level interaction latency associated with traditional adaptive trees originally designed for static settings, thus enabling data-adaptive structures for streaming data.}

\item We conduct extensive experiments on four real-world datasets.
  Experimental results demonstrate that MTSP-LDP consistently outperforms
  state-of-the-art approaches designed for these specific tasks, achieving high
  utility under strict privacy guarantees.
\end{itemize}


%% file: samples/sigmod2026/preliminary.tex
\section{Preliminaries and Problem Definition}
\label{sec:Preliminaries}

\input{samples/sigmod2026/ss_data_stream}

\input{samples/sigmod2026/ss_w-event_ldp}

\input{samples/sigmod2026/ss_fo}

\input{samples/sigmod2026/ss_hio}

\input{samples/sigmod2026/ss_query_definations}

\input{samples/sigmod2026/ss_recent_methods}

%% file: samples/sigmod2026/ss_data_stream.tex
\subsection{Data Stream Model}
\label{ss:data_stream}

We consider a distributed system consisting of a server and a set of $n$ users, $U=\{u_1,\ldots,u_n\}$.
Each user $u_i\in U$ reports a value $v_{it}\in\Omega$ at each discrete timestamp
$t$, e.g., the user's location in a city.
We assume users' values are from a finite domain of size $d$, denoted by
$\Omega=\{\omega_1,\ldots, \omega_d\}$.
All users' reported values at time $t$ form a dataset $D_t$, which consists of
$n$ rows and $d$ columns with $D_t[i][j]=1$ if user $u_i$ reported value
$\omega_j$ at time $t$ and $D_t[i][j]=0$ otherwise.
We allow a user $u_i$ not to report her value at time $t$, i.e., $u_i$ is
inactive at $t$, and in this case, the $i$-th row of $D_t$ is zero.
Let $n_t$ denote the number of active users at time $t$.
Given $D_t$, the server can compute statistics of interest to publish.
For example, the server may continuously publish a vector
$\bc_t=[c_{1t},\ldots,c_{dt}]^\top$ at each time $t$, where $c_{jt}$ denotes the
number of users reporting the location $\omega_j\in\Omega$ at time $t$, and
publishing $\bc_t$ is useful for people to know the traffic congestion situation
in a city at time $t$.

However, the user's value may be private (e.g., her location), and publishing
statistics computed from this raw data can compromise privacy.
For example, suppose that there is only one active user in the system. In this case, the
published statistics directly reveal the user's private value.
In addition, the server is not assumed to be trusted (e.g., it is compromised
or contains software bugs), and it may leak users' private values.
Therefore, privacy-preserving techniques are required to protect users'
sensitive data from the server.

%% file: samples/sigmod2026/ss_w-event_ldp.tex
\subsection{\texorpdfstring{$w$-event LDP on Data Streams}{w-event LDP on Data Streams}}
To precisely define privacy on data streams, we review a useful concept called
$w$-event LDP~\cite{Ren2022}.
Let $V_{it}=(v_{i1},\ldots,v_{it})$ denote user $u_i$'s stream by time $t$.

\begin{definition}[$w$-neighboring]
  Let $V_{it}$ and $V_{it}'$ denote two streams of a single user $u_i$ by time
  $t$.
  Let $w$ denote a positive integer.
  $V_{it}$ and $V_{it}'$ are $w$-neighboring, if for each $v_{ir}$, $v_{is}$,
  $v_{ir}'$, $v_{is}'$ with $1\le r \le s \le t, v_{ir}\neq v_{ir}'$ and $v_{is}
  \neq v_{is}'$, it holds that $s - r + 1 \le w$.
\end{definition}

In other words, we consider two versions of a user $u_i$'s stream $V_{it}$ and
$V_{it}'$, and we say $V_{it}$ and $V_{it}'$ are $w$-neighboring if they have
different values at timestamps fitting in a window of size up to $w$.

\begin{definition}[$w$-event LDP]
  Let $\CM$ be a mechanism that takes as input a stream $V_{it}$ of user $u_i$.
  Let $\CO$ denote the set of all possible outputs of $\CM$.
  We say that $\CM$ satisfies $w$-event $\epsilon$-local differential privacy
  (or, simply, $w$-event LDP) if for any $w$-neighboring streams $V_{it},
  V_{it}'$, $\forall O\subseteq\CO$ and all $t$, it holds that
  \[
    \Pr[\CM(V_{it})\in O] \leq e^\epsilon \Pr[\CM(V_{it}') \in O].
  \]
\end{definition}

The definition captures that a $w$-event LDP mechanism guarantees $\epsilon$-LDP
for each user within any window of size $w$.
The parameter $\epsilon$ controls the privacy risk by injecting a proper amount
of noise to user data~\cite{Dwork2013}.
A smaller $\epsilon$ indicates stronger privacy protection but also lowers the
accuracy and utility of the published result as more noise is added to user
data.

%% file: samples/sigmod2026/ss_fo.tex
\subsection{Frequency Oracle (FO) Under LDP}

Frequency oracle (FO) protocols are common building blocks of many
privacy-preserving techniques.
An FO can estimate the frequency distribution of a private attribute while
preserving privacy.
Here we briefly introduce \emph{optimized unary encoding} (OUE)~\cite{Wang2017l},
a widely used FO protocol that achieves $\epsilon$-LDP with high estimation
accuracy.
OUE consists of the following steps.
\begin{itemize}
\item{\textbf{Encoding.}} 
Each user $u_i$ encodes her private value $v_i\in\Omega$ into a one-hot binary
vector $\bx_i$ of length $d$, where the $j$-th bit $\bx_i[j] = 1$ if and only if
$v_i= \omega_j$, and $0$ otherwise.
\item{\textbf{Perturbation.}}
Instead of directly submitting $\bx_i$ to a server, OUE perturbs each bit of
$\bx_i$ independently.
Specifically, if $\bx_i[j]=1$, it remains $1$ with probability $p = 1/2$ and
flips to $0$ with probability $1-p$.
If $\bx_i[j]=0$, it flips to $1$ with probability $q = \frac{1}{e^\epsilon + 1}$
and remains $0$ with probability $1-q$.
The perturbed vector $\bx_i'$ is then sent to the server, thereby preserving
users' privacy.
\item{\textbf{Aggregation.}}
The server aggregates the perturbed vectors from all users and estimates the
frequency for each value $\omega_j$.
Let $\by[j]$ be the total number of perturbed reports with bit $j$ set to $1$.
An unbiased estimate of the true frequency $f_j$ is given by
\[
  \hat{f}_j = \frac{\by[j] - nq}{n(p - q)}, \quad j=1,\ldots,d.
\]
\item{\textbf{Estimation Error.}}
OUE satisfies $\epsilon$-LDP, and the variance of the OUE estimator $\hat{f}_j$
is
\[
  \var(\hat{f}_j) = \frac{4e^\epsilon}{n(e^\epsilon - 1)^2}, \quad j=1,\ldots,d.
\]
This variance depends on the number of users $n$ and the privacy budget
$\epsilon$.
To simplify notation, in the following discussion, we will denote the variance
of OUE by $\var(\epsilon)$ while the number of users should be clear from
context.
\end{itemize}

%% file: samples/sigmod2026/ss_hio.tex
\subsection{Private Binary Tree for Static Data}
\label{ss:tree}

Private binary trees are widely used in many scenarios for different
purposes under LDP~\cite{Chan2011, Honaker2015, wang2019HIO}.
For example, Chan et al.~\cite{Chan2011} first utilized a binary interval tree
to represent a binary stream where each incoming binary value is assigned to a
leaf node, enabling the efficient computation of time-range queries.
Wang et al.~\cite{wang2019HIO} leveraged a hierarchy interval tree to support
multi-dimensional queries in the static setting.

In this work, we 
use this data structure as summary statistics of user data $D_t$. 
As illustrated in \cref{fig:tree}, a private binary tree is a perfect binary
tree defined on the value domain $\Omega$.
Each node of the tree is associated with an interval (or a set), that is the
union of its child nodes' intervals.
Nodes in the same level have disjoint intervals, and they form a partition of
$\Omega$.
Each node is assigned with a \emph{property}, e.g., the fraction of users
holding a value in the node's interval (i.e., its frequency).
Formally, let $\DT_t$ denote the tree built from user data $D_t$ at time $t$,
and let $\DT_t[a]$ denote the property of node $a$.
For example, $\DT_t[a]$ in \cref{fig:tree} denotes the fraction of users holding
a value in set $\{0,1\}$.

\begin{figure}[htp]
  \centering
  \input{samples/sigmod2026/tz_hio}
  \caption{Representing user data as a private binary tree, using FO to estimate
    each node's property in LDP, and performing a range query on interval
    $[0,6]$.}
  \label{fig:tree}
  \Description{The figure illustrates the data structure and query processing flow. 
On the left, a data table labeled $D_t$ shows $n$ users (rows $u_1$ to $u_n$) with one-hot encoded values across 8 columns ($0$ to $7$). 
An arrow points to the center, where a binary tree structure represents the domain decomposition. 
The tree shows how a range query for the interval $[0, 6]$ is decomposed into a minimal set of nodes: a grey node labeled 'e' covering $\{0,...,3\}$, a grey node labeled 'c' covering $\{4,5\}$, and a grey leaf node labeled 'f' covering $\{6\}$. 
Other nodes like 'a', 'b', and 'd' are unshaded. 
On the right, three groups of users provide input to three Frequency Oracles ($FO_1$, $FO_2$, $FO_3$), which contribute data to different levels of the tree structure.}
\end{figure}
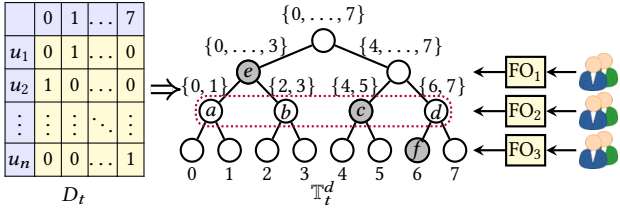

While in the LDP setting, the server can not directly access user data $D_t$.
To build a private binary tree of $D_t$, the server can leverage an FO.
Specifically, we randomly partition users into $L-1$ disjoint groups of equal
size, where $L$ is the number of levels of the tree.
For each level except level $0$ (which only has the root node representing
domain $\Omega$), the server invokes an FO whose domain size equals the number
of nodes at that level, using privacy budget $\epsilon$.
Then, for each node in the level, the frequency estimate of FO is actually the
fraction of users holding a private value in the node's interval.
Because each user participates at exactly one level, by parallel composition the
overall procedure still satisfies $\epsilon$-LDP.
For example, in \cref{fig:tree}, the server can invoke $\FO_2$ in level $2$ on a
domain with $4$ intervals, and the frequency estimates of $\FO_2$ are estimates
of $\DT_t[x]$ for $x\in\{a,b,c,d\}$.

In the following discussion, we let $\hDT_t$ denote the estimated private binary
tree of $D_t$.
Note that $\hDT_t$ and $\DT_t$ have the same structure except that node
properties in $\hDT_t$ are estimated using an FO.
Because every FO is run for each level with the same number of users, the
variance of $\hDT_t[a]$ is the same, i.e., $\var(\epsilon)$ when privacy budget
is $\epsilon$.

The estimated private binary tree $\hDT_t$ can be used to answer range queries
with high estimation accuracy.
To answer a range query for a specified interval, we select a minimum number of
nodes in the tree that cover the query interval.
Node properties of this minimum cover can be used to answer the query.
For example, in \cref{fig:tree}, if we want to estimate how many users holding a
private value in $\{0,\ldots,6\}$, we first find a minimum cover $\{e,c,f\}$ and
then answer the query as $(\hDT_t[e]+\hDT_t[c]+\hDT_t[f])n_t$.

While this structure effectively supports queries on static datasets, its direct extension to infinite streams entails significant limitations. A natural approach is to treat each time step independently and build a separate private tree for each $D_t$. However, this ignores the temporal correlations inherent in streaming data and requires composing the privacy loss over a long (potentially unbounded) sequence of releases, which leads to poor utility under our streaming privacy notion. Moreover, the tree partition is fixed and cannot adapt to time-varying data distributions; when the distribution changes over time, many nodes become either too sparse or too dense, resulting in large estimation errors. These limitations necessitate the development of our proposed framework, which builds on the above tree representation but incorporates dynamic budget allocation and data-adaptive mechanisms to handle evolving streams.


%% file: samples/sigmod2026/tz_hio.tex
\begin{tikzpicture}[thick,font=\footnotesize,>=stealth,
  every node/.style={inner sep=1pt},
  nd/.style={circle,draw, minimum size=3mm, inner sep=0pt},
  arr/.style={-Implies,double,double distance=2pt},
  rec/.style={draw,minimum width=5mm,minimum height=4mm,fill=yellow!20},
  group/.style={draw,densely dotted,rounded corners=1.5mm},
  ]

  \coordinate (O);
  \foreach \i in {0,...,7}{
    \node[nd,right=\i*0.5 of O,label=below:{$\i$}] (n3\i) {};
  }

  \path (n30)--(n31) node[midway,above=.35,nd] (n20) {$a$};
  \path (n32)--(n33) node[midway,above=.35,nd] (n21) {$b$};
  \path (n34)--(n35) node[midway,above=.35,nd] (n22) {};
  \path (n36)--(n37) node[midway,above=.35,nd] (n23) {$d$};
  \path (n20)--(n21) node[midway,above=.35,nd] (n10) {};
  \path (n22)--(n23) node[midway,above=.35,nd] (n11) {};
  \path (n10)--(n11) node[midway,above=.25,nd] (n00) {};

  \draw (n11) -- (n00) -- (n10);
  \draw (n20) -- (n10) -- (n21);
  \draw (n22) -- (n11) -- (n23);
  \draw (n30) -- (n20) -- (n31);
  \draw (n32) -- (n21) -- (n33);
  \draw (n34) -- (n22) -- (n35);
  \draw (n36) -- (n23) -- (n37);

  \node[above=0 of n00] {$\{0,\ldots,7\}$};
  \node[above=0 of n10.100] {$\{0,\ldots,3\}$};
  \node[above=0 of n11.80] {$\{4,\ldots,7\}$};
  \node[above=0 of n20.120] {$\{0,1\}$};
  \node[above=0 of n21.60] {$\{2,3\}$};
  \node[above=0 of n22.120] {$\{4,5\}$};
  \node[above=0 of n23.60] {$\{6,7\}$};

  \coordinate[below=.35 of n30] (s);


    \path (n30) -- (n37) node[midway, below=0.35cm] {$\DT_t$};
    
  \node[nd,fill=gray!50] at (n10) {$e$};
  \node[nd,fill=gray!50] at (n22) {$c$};
  \node[nd,fill=gray!50] at (n36) {$f$};

  \begin{scope}[xshift=-1.4cm,yshift=.8cm]
    \node[label=below:{$D_t$},inner sep=2pt] (dat) {%
      $\begin{tblr}{
        hlines,vlines,colsep=3pt,rowsep=1pt,
        colspec={ccccc},
        column{1}={colsep=1pt,blue!10},
        row{4}={rowsep=0pt},
        column{4}={colsep=0pt},
        row{1}={blue!10},
        cell{2-5}{2-5}={yellow!20},
        }
               & 0      & 1      & \dots  & 7      \\
        u_1    & 0      & 1      & \dots  & 0      \\
        u_2    & 1      & 0      & \dots  & 0      \\
        \vdots & \vdots & \vdots & \ddots & \vdots \\
        u_n    & 0      & 0      & \dots  & 1      \\
      \end{tblr}$
    };

    \draw[arr] (dat.east) -- ++(.4,0);
  \end{scope}

  \coordinate[right=2.5 of n11] (A);
  \node (u1) at (A) {\includegraphics[width=5mm]{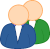}};
  \node (u2) at (n23 -| A) {\includegraphics[width=5mm]{figs/users}};
  \node (u3) at (n37 -| A) {\includegraphics[width=5mm]{figs/users}};

  \foreach \i in {1,2,3}{
    \node[rec,left=.4 of u\i] (fo\i) {FO\textsubscript{\i}};
    \draw[->] (u\i) -- (fo\i);
    \draw[->] (fo\i) -- ++(-.7,0);
  }

  \node[group,purple,fit={(n20) (n23)}] (g2) {};

\end{tikzpicture}

%% file: samples/sigmod2026/ss_query_definations.tex
\subsection{Queries on Data Streams}
\label{ss:queries}

Given a data stream formed by user data at each timestamp, i.e., $\CD =
(D_1,D_2,\ldots)$, other than the histogram/frequency query extensively studied
in previous work~\cite{Kellaris2014, Ren2022}(i.e., the query $\bc_t/n_t$ defined in \cref{ss:data_stream}), we
define several novel streaming query tasks.

\begin{definition}[Counting Query]
  Let $c_t(v)$ denote the number of users reporting the private value $v\in\Omega$ at
  time $t$, i.e., $c_t(v)\triangleq |\{u_i\in U\colon v_{it}=v\}|$.
  Given a value $v\in\Omega$ and a positive integer $\Delta$, a counting query
  aims to report the total number of times users report the private value $v$ in
  the most recent $\Delta$ timestamps.
  Formally,
  \[
    Q_t^c(v,\Delta)\triangleq \sum_{t'=t-\Delta +1}^t c_{t'}(v).
  \]
\end{definition}

The counting query is a generalization of the histogram query and has widespread
applications.
Consider the following example.

\begin{example}\label{exam:vehicle}
  Consider a vehicle dispatch management system in a city.
  At each time $t$, the system needs to publish the total number of times
  vehicles passed by a location $v$ over the past ten minutes.
  Assuming a time granularity of one minute, then the query corresponds to
  $Q_t^c(v, 10)$.
\end{example}

As a special case, when $\Delta=1$, $Q_t^c(v,1)=c_t(v)$ which is the number of
users reporting $v$ at time $t$, and $Q_t^c(v,1)/n_t$ becomes the frequency of
value $v$ at time $t$.

\begin{definition}[Range Query]
  Given a value range $[v_1, v_2]$ (where $v_1\leq v_2$) and a positive integer
  $\Delta$, a range query aims to report at each time $t$ the total number of
  times users report private values in $[v_1, v_2]$ in the most recent
  $\Delta$ timestamps.
  Formally,
  \[
    Q_t^r([v_1, v_2],\Delta)\triangleq
    \sum_{t'=t-\Delta +1}^t\sum_{v\in[v_1,v_2]} c_{t'}(v)
  \]
\end{definition}

The range query further generalizes the counting query by allowing the user value to
be in a specified range.
Consider the following example.

\begin{example}\label{exam:disease}
  A disease monitoring center needs to publish daily the total number of
  infected people with age above $60$ over the past week during the COVID
  pandemic.
  Assuming a time granularity of one day, then the query corresponds to $Q_t^r([60,
  \infty), 7)$.
\end{example}

In many real-world scenarios, the server also wants to identify specific
abnormal events in the stream through a class of queries known as event
monitoring.

\begin{definition}[Event Monitoring]\label{def:event_monitoring}
  Let $\alpha$ and $\beta$ be two functions, where $\alpha$ could be a counting
  query or a range query on the stream, and $\beta$ is a Boolean function defined
  on the output of $\alpha$.
  An event monitoring query, denoted by $Q_t^e(\alpha,\beta)$, aims to report at
  each time $t$ a binary result where $1$ denotes the monitored event occurs,
  and $0$ otherwise.
\end{definition}

We illustrate the usefulness of this concept by the following example.
Let $\Delta>0$ denote a specified monitoring period, and we are interested in
monitoring whether counts of a value $v$ significantly increase in two
consecutive monitoring periods.
Then the two functions $\alpha$ and $\beta$ can be formally defined as
\begin{align*}
  \alpha(v,\Delta) &= Q_t^c(v, \Delta) - Q_{t-w+1}^c(v, \Delta), \\
  \beta(x,\vartheta) &= \indr{x>\vartheta},
\end{align*}
where $\indr{B}$ denotes an indicator function, and it returns $1$ if the
condition $B$ is true, otherwise $0$.
The event monitoring query is a compound function on the stream, i.e.,
$Q_t^e(\alpha,\beta)= \beta\circ\alpha$, which reports $1$ if there is a sudden
increase of counts of value $v$ in two consecutive monitoring periods.
It is also straightforward to define $\alpha$ using a range query to monitor
multiple user values.
Therefore, we can use the defined event monitoring queries to monitor the sudden
change in counts in \cref{exam:vehicle,exam:disease}.

\paragraph{Remarks}
The streaming queries proposed above are generalizations of existing
works~\cite{Kellaris2014,Chen2017} with the emphasis on data streams and aiming
to achieve $w$-event LDP.
Furthermore, for parameter $\Delta$ in these queries, we require $\Delta \leq w$
in order to satisfy $w$-event LDP.

%% file: samples/sigmod2026/ss_recent_methods.tex
\subsection{\texorpdfstring{Recent Methods to Achieve $w$-event LDP}{Recent Methods to Achieve w-event LDP}}
\label{ss:competitors}

State-of-the-art $w$-event LDP methods~\cite{Ren2022} focus on releasing
statistical histograms at each timestamp.
These methods can be categorized into the following four types based on their
privacy budget allocation strategies.
\begin{itemize}
\item {LDP Budget Uniform (LBU)} applies a fixed $\epsilon/w$ budget to each
timestamp in the window.
However, as window size $w$ increases, the allocated budget per timestamp
becomes vanishingly small, introducing excessive noise and severely impairing
utility.

\item{LDP Sampling (LSP)} allocates the entire privacy budget $\epsilon$ to a
single timestamp within the sliding window, providing high accuracy at that
point while reusing its result to approximate the data for all other timestamps.
This method performs well when the stream is stable but fails to adapt when the
stream fluctuates.
The approximation errors can become very large if subsequent stream data differ
significantly from previous releases.

\item{LDP Budget Distribution (LBD)} and \textbf{LDP Budget Absorption
  (LBA)} adopt dynamic budget allocation techniques, drawing inspiration from
$w$-event CDP~\cite{Kellaris2014}.
Both methods consist of two sub-mechanisms: private dissimilarity estimation and
private strategy determination.
The private dissimilarity estimation sub-mechanism computes the private
dissimilarity between the current true statistics and the previous release based
on perturbed user data collected via an FO with a fixed budget
$\epsilon_{t,1}=\epsilon/(2w)$.
The private strategy determination sub-mechanism decides whether to publish new
statistics or reuse previous values by comparing the estimated dissimilarity
with a potential publication error determined by the publication budget
$\epsilon_{t,2}$.
The allocation of $\epsilon_{t,2}$ differs between LBD and LBA.
In LBD, the budget is distributed across timestamps requiring data publication
in an exponentially decaying manner, and budget spent at timestamps outside the
current window is reclaimed for reuse.
In contrast, LBA first allocates the budget uniformly and then absorbs it at the
timestamps that use approximation.


\item{Population Division Extensions.}
To improve the utility of LDP mechanisms, population division has emerged as a
promising strategy~\cite{Kulkarni2019,Wang2017l,Ren2022}.
Instead of splitting the privacy budget across timestamps, this paradigm
partitions the user population, assigning each user to report at a specific
timestamp using the entire privacy budget.
This reduces estimation variance via the amplification-by-subsampling effect and
enables more accurate data release within each reporting round.
Based on this idea, LBU, LBD, and LBA can be extended to LPU (Population-based
Uniform), LPD (Population-based Distribution), and LPA (Population-based
Absorption), respectively.
Notably, LSP intrinsically follows this paradigm.
Population division is a general augmentation strategy that can be applied to
any FO-based mechanism, including our proposed MTSP-LDP.
While population division improves utility within each reporting round, it does
not capture temporal correlations across timestamps, which are critical in
streaming scenarios.
\end{itemize}
\paragraph{Remarks}
Adaptive methods such as LBD and LBA, while designed to flexibly allocate
privacy budgets based on data variations, have been shown to exhibit even lower
utility in practice\cite{Schaler2024}, failing to leverage the potential
advantages of dynamic privacy budget allocation.
Specifically, the private dissimilarity estimation sub-mechanism consumes half
of the total privacy budget to collect perturbed data.
This data is then discarded after estimating dissimilarity, halving the budget
available for data publication.
Additionally, these methods rely solely on single-timestamp information,
neglecting the rich temporal correlations within the sliding window.
Consequently, this strategy is insufficient, as the $w$-event LDP constraint
requires the total budget $\epsilon$ to be allocated reasonably across all $w$
timestamps, not just based on single-timestamp changes and the instantaneous
budget cap.

%% file: samples/sigmod2026/method.tex
\section{The MTSP-LDP Framework}
\label{sec:method}

In this section, we propose a novel framework for \textbf{M}ulti-\textbf{T}ask
\textbf{S}treaming data \textbf{P}ublication with $w$-event \textbf{LDP}
guarantee (MTSP-LDP).
MTSP-LDP is designed to efficiently answer streaming queries on infinite data
streams and achieve $w$-event LDP.

\input{samples/sigmod2026/ss_overview}


\input{samples/sigmod2026/ss_dissimilarity}

\input{samples/sigmod2026/ss_budget_allocation}

\input{samples/sigmod2026/ss_tree_release}

\input{samples/sigmod2026/ss_query_processing}

%% file: samples/sigmod2026/ss_overview.tex
\subsection{Overview}

At a high level, MTSP-LDP follows the \emph{dissimilarity guided publication}
framework, first proposed for the $w$-event CDP setting~\cite{Kellaris2014} and
later extended to the $w$-event LDP setting~\cite{Ren2022}.
The main idea of this framework is that, the server checks at every timestamp
whether it is more beneficial to approximate the current stream statistics with
the last released statistics, than to publish newly computed stream statistics
with necessary noise.
MTSP-LDP further enhances this framework by proposing an \emph{optimal budget
  allocation strategy} to improve estimation accuracy and a \emph{private binary
  tree structure} to enable multi-task streaming queries.
In more detail, MTSP-LDP consists of four mechanisms, i.e., private
dissimilarity estimation, optimal privacy budget allocation, private adaptive
tree publication, and budget-free multi-task streaming query, as illustrated in
\cref{fig:framework}.

\begin{figure}[htp]
  \centering
  \input{samples/sigmod2026/tz_framework}
  \caption{The framework of MTSP-LDP}
  \label{fig:framework}
  \Description{The figure illustrates the workflow of the MTSP-LDP framework at timestamp $t$. At the top, $n_t$ active users send perturbed data to a Frequency Oracle. The process then flows through four connected modules labeled $M_1^t$ through $M_4^t$.}
\end{figure}
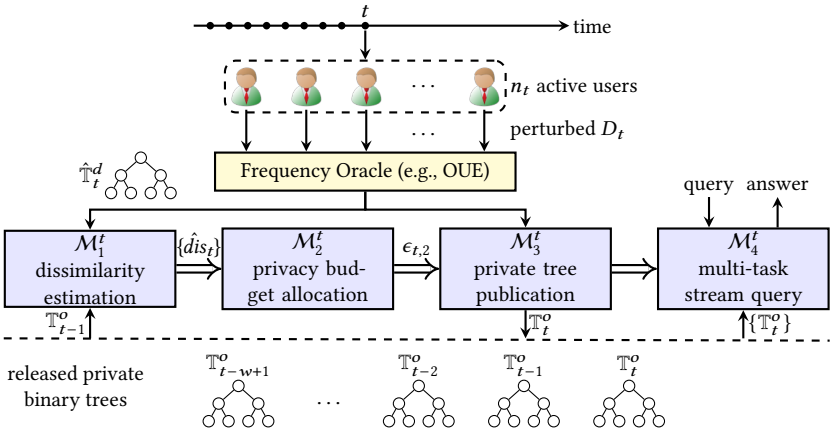
\begin{itemize}
\item {\textbf{Private Dissimilarity Estimation $\CM_1^t$}}.
At each time $t$, the server computes the dissimilarity $\dis_t$ between the
statistics of current user data $D_t$ and the last released stream statistics,
and $\dis_t$ will be used to determine whether to approximate with the previous
release or publish with noise.
The main challenge we need to address in $\CM_1^t$ is how to privately compute
$\dis_t$ as user data $D_t$ is not available to the server in LDP, and instead,
we design an unbiased estimator $\hdis_t$ to estimate $\dis_t$ while preserving
user privacy.
\item {\textbf{Optimal Privacy Budget Allocation $\CM_2^t$}}.
Given $\hdis_t$, existing methods~\cite{Kellaris2014,Ren2022} directly compare
$\hdis_t$ with the publication error $\err_t$. If $\hdis_t < \err_t$, then they
approximate with the previous release; otherwise, they publish with noise using a
predetermined privacy budget.
However, we notice that this commonly used strategy is often not optimal,
particularly in the $w$-event setting, as it bases its decision solely on the
current timestamp while neglecting the temporal context within the sliding window.
We therefore propose an optimal privacy budget allocation strategy by cleverly leveraging
previously computed dissimilarities, i.e., $\{\hdis_{t'}\}_{t'\leq t}$.
This strategy will determine whether to approximate or to publish, and how to
optimally allocate the privacy budget at time $t$.

\item {\textbf{Private Adaptive Tree Publication $\CM_3^t$}}.
In MTSP-LDP, the server privately collects user data $D_t$ and represents $D_t$
as a private binary tree to both protect user privacy and support
multi-task streaming query.
This private binary tree participates in dissimilarity estimation in
$\CM_1^t$ as well as streaming query in $\CM_4^t$.
We construct the private binary tree in a data-adaptive way in order to
accurately capture the statistics of user data $D_t$.
The output tree $\OT_t$ is then released and will be used by $\CM_4^t$ to run
streaming queries.

\item {\textbf{Budget-Free Multi-task Streaming Query $\CM_4^t$}}.
Unlike existing research, which typically focuses on ad-hoc queries, e.g., counting, histogram query, or frequency query,
MTSP-LDP is designed to support multi-task streaming queries, including all
queries defined in \cref{ss:queries}, by utilizing the released private binary
trees $\{\OT_{t'}\}_{t'\leq t}$ without incurring  any additional privacy budget.
We theoretically show that MTSP-LDP satisfies $w$-event LDP.
\end{itemize}

These above sub-mechanisms work in tandem to achieve privacy-preserving data analysis in streaming environments, as illustrated in Algorithm~\ref{alg:mtsp-ldp}. Next, we
will introduce them in detail.
\input{samples/sigmod2026/alg_method}

%% file: samples/sigmod2026/tz_framework.tex
\begin{tikzpicture}[>=stealth,thick,font=\small,
  every node/.style={inner sep=1pt,font=\footnotesize},
  group/.style={draw,dashed,rounded corners=1.5mm},
  rec/.style={draw,minimum height=1cm,text width=2.2cm,align=center,fill=blue!10},
  arr/.style={-Implies,double,double distance=2pt},
  point/.style={circle,fill,minimum size=2pt},
  nd/.style={circle,draw,minimum size=1.5mm,inner sep=0pt}
  ]

  \coordinate (O);
  \draw[->] (O) ++(-.7,0) -- ++(5,0) node[right] {time};

  \node (u1) at (0,-.8) {\includegraphics[width=4mm]{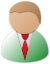}};
  \node[right=.3 of u1] (u2) {\includegraphics[width=4mm]{figs/user}};
  \node[right=.3 of u2] (u3) {\includegraphics[width=4mm]{figs/user}};
  \node[right=.3 of u3] (dots) {$\dots$};
  \node[right=.3 of dots] (un) {\includegraphics[width=4mm]{figs/user}};

  \node[group,fit={(u1) (un)},label=right:{$n_t$ active users}] (users) {};

  \node[point,label=above:{$t$}] (t) at (users.north |- O) {};
  \foreach \i in {1,...,8}{
    \node[point,left=.25*\i of t.center] {};
  }

  \draw[->] (t) -- (users.north);

  \node[draw,fill=yellow!20,minimum width=4cm,minimum height=.5cm,
  below=.5 of users] (FO) {Frequency Oracle (e.g., OUE)};

  \foreach \i in {1,...,3}{
    \draw[->] (u\i) -- (u\i |- FO.north);
  }
  \draw[->] (un) -- (un |- FO.north) node[midway,right=2ex] {perturbed $D_t$};
  \node[below=.5 of dots] {$\dots$};

  \node[rec,below left=.5 and .5 of FO] (M1) {$\CM_1^t$\\dissimilarity estimation};
  \node[rec,right=.6 of M1] (M2) {$\CM_2^t$\\privacy bud-get allocation};
  \node[rec,right=.6 of M2] (M3) {$\CM_3^t$\\private tree publication};
  \node[rec,right=.6 of M3] (M4) {$\CM_4^t$\\multi-task stream query};

  \draw[arr] (M1) -- (M2) node[midway,above=3pt] {$\{\!\hdis_t\!\}$};
  \draw[arr] (M2) -- (M3) node[midway,above=3pt] {$\epsilon_{t,2}$};
  \draw[arr] (M3) -- (M4);

  \draw[->] (FO) -- ++(0,-.5) -| (M1.north);

  \draw[->] (FO) -- ++(0,-.5) -| (M3.north);

  \begin{scope}[xshift=-1.8cm,yshift=-2.2cm,thin]
    \node[nd] (n31) {};
    \node[nd,right=.1 of n31] (n32) {};
    \node[nd,right=.1 of n32] (n33) {};
    \node[nd,right=.1 of n33] (n34) {};

    \path (n31)--(n32) node[midway,above=.15,nd] (n21) {};
    \path (n33)--(n34) node[midway,above=.15,nd] (n22) {};
    \path (n21)--(n22) node[midway,above=.15,nd] (n1) {};

    \draw (n21) -- (n1) -- (n22);
    \draw (n31) -- (n21) -- (n32);
    \draw (n33) -- (n22) -- (n34);

    \node[left=.1 of n21] {$\hDT_t$};
  \end{scope}

  \begin{scope}[xshift=4.7cm,yshift=-5.2cm,thin]
    \node[nd] (n31) {};
    \node[nd,right=.1 of n31] (n32) {};
    \node[nd,right=.1 of n32] (n33) {};
    \node[nd,right=.1 of n33] (n34) {};

    \path (n31)--(n32) node[midway,above=.15,nd] (n21) {};
    \path (n33)--(n34) node[midway,above=.15,nd] (n22) {};
    \path (n21)--(n22) node[midway,above=.15,nd] (n1) {};

    \draw (n21) -- (n1) -- (n22);
    \draw (n31) -- (n21) -- (n32);
    \draw (n33) -- (n22) -- (n34);

    \node[above=0 of n1] {$\OT_t$};
  \end{scope}

  \begin{scope}[xshift=3.3cm,yshift=-5.2cm,thin]
    \node[nd] (n31) {};
    \node[nd,right=.1 of n31] (n32) {};
    \node[nd,right=.1 of n32] (n33) {};
    \node[nd,right=.1 of n33] (n34) {};

    \path (n31)--(n32) node[midway,above=.15,nd] (n21) {};
    \path (n33)--(n34) node[midway,above=.15,nd] (n22) {};
    \path (n21)--(n22) node[midway,above=.15,nd] (n1) {};

    \draw (n21) -- (n1) -- (n22);
    \draw (n31) -- (n21) -- (n32);
    \draw (n33) -- (n22) -- (n34);

    \node[above=0 of n1] {$\OT_{t-1}$};
  \end{scope}

  \begin{scope}[xshift=1.9cm,yshift=-5.2cm,thin]
    \node[nd] (n31) {};
    \node[nd,right=.1 of n31] (n32) {};
    \node[nd,right=.1 of n32] (n33) {};
    \node[nd,right=.1 of n33] (n34) {};

    \path (n31)--(n32) node[midway,above=.15,nd] (n21) {};
    \path (n33)--(n34) node[midway,above=.15,nd] (n22) {};
    \path (n21)--(n22) node[midway,above=.15,nd] (n1) {};

    \draw (n21) -- (n1) -- (n22);
    \draw (n31) -- (n21) -- (n32);
    \draw (n33) -- (n22) -- (n34);

    \node[above=0 of n1] {$\OT_{t-2}$};
    \node[left=.6 of n21] {$\dots$};
  \end{scope}

  \begin{scope}[xshift=-.5cm,yshift=-5.2cm,thin]
    \node[nd] (n31) {};
    \node[nd,right=.1 of n31] (n32) {};
    \node[nd,right=.1 of n32] (n33) {};
    \node[nd,right=.1 of n33] (n34) {};

    \path (n31)--(n32) node[midway,above=.15,nd] (n21) {};
    \path (n33)--(n34) node[midway,above=.15,nd] (n22) {};
    \path (n21)--(n22) node[midway,above=.15,nd] (n1) {};

    \draw (n21) -- (n1) -- (n22);
    \draw (n31) -- (n21) -- (n32);
    \draw (n33) -- (n22) -- (n34);

    \node[above=0 of n1] {$\OT_{t-w+1}$};
  \end{scope}

  \coordinate[below=.4 of M1.south west] (A);
  \coordinate[below=.4 of M4.south east] (B);

  \draw[dashed] (A) -- (B);

  \draw[<-] (M1.south) -- (M1 |- A) node[midway,left] {$\OT_{t-1}$};
  \draw[->] (M3.south) -- (M3 |- A) node[midway,right] {$\OT_t$};
  \draw[<-] (M4.south) -- (M4 |- A) node[midway,right] {$\{\OT_t\}$};

  \draw[<-] (M4.130) -- ++(0,.4) node[above, text depth=.25ex] {query};
  \draw[->] (M4.50) -- ++(0,.4) node[above, text depth=.25ex] {answer};

  \node[below=2ex of A,anchor=north west,align=center] {released private \\ binary trees};
\end{tikzpicture}

%% file: samples/sigmod2026/alg_method.tex
\begin{algorithm}[H]
  \caption{MTSP-LDP}
  \label{alg:mtsp-ldp}
  \KwIn{Data stream $\CD=(D_1,D_2,\ldots)$, privacy budget $\epsilon$,
    sliding window size $w$}
  \KwOut{Respond to queries $Q_t^c, Q_t^r$ or $Q_t^e$ at each time $t$}

  \ForEach{$t=1,2,\ldots$}{
    \eIf{$t \leq w$}{
      $\OT_t\gets\FO(D_t, \epsilon/w)$\tcp*{run FO on $D_t$ with budget
        $\epsilon/w$}
    }{
      \tcp{1. private dissimilarity estimation}
      $\hDT_t\gets\FO(D_t, \epsilon/(2w))$\;
      $\hdis_t\gets\text{dissimilarity}(\hDT_t, \OT_{t-1})$\;
      \tcp{2. optimal privacy budget allocation}
      $\epsilon_{t,2}\gets$ run $\CM_3^t(\{\hdis_t\})$\;
      \tcp{3. private tree publication}
      \eIf{$\epsilon_{t,2}>0$}{
        $\hAT_t\gets\texttt{ATC}(\hDT_t)$\tcp*{adaptive tree construction}
      }{
        $\hAT_t\gets\OT_{t-1}$\;
      }
      $\OT_t\gets\texttt{GroupSmooth}(\hAT_t)$\;
    }
  }
  \tcp{4. multi-task processing}
  Use $\{\OT_{t'}\}_{t'\leq t}$ to respond to queries $Q_t^c, Q_t^r$ or $Q_t^e$.
\end{algorithm}

%% file: samples/sigmod2026/ss_dissimilarity.tex
\subsection{Private Dissimilarity Estimation}

We are now ready to describe each mechanism in MTSP-LDP.
Recall that MTSP-LDP follows the dissimilarity guided publication framework,
where the server checks at every timestamp whether it is more beneficial to
approximate the current statistics with the last released statistics, than to
publish new statistics with noise.
This requires the server to compute the dissimilarity between user data $D_t$ at
time $t$ and the last released statistics.
It will be clear later that the last released statistics is also represented as
a private binary tree, denoted by $\OT_{t-1}$ (see \cref{ss:tree_publication}).
To compare the dissimilarity between $D_t$ and $\OT_{t-1}$, we can convert $D_t$
to a private binary tree $\DT_t$ using the method described in \cref{ss:tree},
and the dissimilarity is defined as the mean of the squared differences between
corresponding node properties in these two trees, i.e.,
\[
  \dis_t = \mathrm{dissimilarity}(\DT_t, \OT_{t-1}) \triangleq
  \frac{1}{|\DT_t|}\sum_{a\in \DT_t}(\DT_t[a] - \OT_{t-1}[a])^2,
\]
where $|\DT_t|$ denotes the number of nodes in the tree.

However, in the LDP setting, the server cannot directly access user data $D_t$,
hence $\DT_t$ is unknown.
Instead, the server can estimate $\DT_t$ by $\hDT_t$ using several FOs, as
explained in \cref{ss:tree}.
Next, we present an unbiased estimator of $\dis_t$.

\begin{theorem}\label{thm:dis}
  If $\hDT_t$ is estimated with privacy budget $\epsilon_{t,1}$, then $\hdis_t$
  is an unbiased estimate of $\dis_t$ and satisfies $\epsilon_{t,1}$-LDP, where
  \[
    \hdis_t\triangleq\frac{1}{|\hDT_t|}
    \sum_{a\in\hDT_t} (\hDT_t[a] - \OT_{t-1}[a])^2 - \var(\epsilon_{t,1}).
  \]
\end{theorem}


\begin{proof}
We begin by computing the expectation of $\hdis_t$:
\[
\E(\hdis_t)
  = \E\big[ \frac{1}{|\hDT_t|}
    \sum_{a\in\hDT_t}(\hDT_t[a]-\OT_{t-1}[a])^2 - \var(\epsilon_{t,1}) \big]
\]
Since $\var(\epsilon_{t,1})$ is a constant, it can be factored out of the expectation:
\begin{align*}
  \E(\hdis_t)
  &= \frac{1}{|\hDT_t|}\sum_{a\in \hDT_t} \E\big[
    (\hDT_t[a] - \OT_{t-1}[a])^2\big]  - \E\big[ \var(\epsilon_{t,1}) \big] \\
  &= \frac{1}{|\hDT_t|}\sum_{a\in \hDT_t} \big[
    (\DT_t[a] - \OT_{t-1}[a])^2 + \var(\hDT_t[a])\big] - \var(\epsilon_{t,1}) \\
  &= \frac{1}{|\hDT_t|}\sum_{a\in \hDT_t} \big[
    (\DT_t[a] - \OT_{t-1}[a])^2 + \var(\hDT_t[a]) - \var(\epsilon_{t,1}) \big]
\end{align*}

Since $\hDT_t[a]$ is generated by the FO mechanism with
budget $\epsilon_{t,1}$, we have $\var(\hDT_t[a])=\var(\epsilon_{t,1})$. Thus, the expectation becomes:
\[
  \E(\hdis_t)
  = \frac{1}{|\hDT_t|}\sum_{a\in \hDT_t} \big[
  (\DT_t[a] - \OT_{t-1}[a])^2 \big]
  = dis_t
\]

Finally, because differential privacy is immune to post-processing~\cite{Dwork2013},
estimator $\hdis_t$ is still $\epsilon_{t,1}$-LDP as long as $\hDT_t$ is
$\epsilon_{t,1}$-LDP.
\end{proof}

\paragraph{Remarks}
The dissimilarity is estimated at every timestamp using a fixed privacy budget
$\epsilon_{t,1}=\epsilon/(2w)$, i.e., half of the total privacy budget
uniformly spent on each timestamp in the window.
Since differential privacy is immune to post-processing~\cite{Dwork2013},
estimator $\hdis_t$ is still $\epsilon_{t,1}$-LDP as long as $\hDT_t$ is
$\epsilon_{t,1}$-LDP.

%% file: samples/sigmod2026/ss_budget_allocation.tex
\subsection{Optimal Privacy Budget Allocation}
\label{ss:optimal_budget}

The server now needs to make a decision at time $t$ whether it should
approximate the current statistics using the last released statistics (referred
to as \emph{approximation}), or publish the current statistics with the necessary
noise (referred to as \emph{publication}).
Both options will introduce errors in the released statistics, i.e., the
approximation incurs error $\hdis_t$, and the publication incurs error
$\var(\epsilon_{t,2})$ which is the estimation error of running FO with privacy
budget $\epsilon_{t,2}$.

Existing methods~\cite{Kellaris2014,Ren2022} directly compare $\hdis_t$ with
$\var(\epsilon_{t,2})$, and if
$\hdis_t<\var(\epsilon_{t,2})$, then the server chooses approximation, otherwise it chooses
publication.
We notice that this commonly used strategy is myopic, focusing only on the current
timestamp and neglecting long-term overall accuracy of the released
statistics.
Recent empirical studies also show that these existing methods often perform
worse than even static baseline methods~\cite{Schaler2024}.
To address this weakness, we propose a novel strategy, i.e.,
\emph{\textbf{O}ptimal privacy \textbf{B}udget \textbf{A}llocation (OBA)}, that
aims to achieve high long-term accuracy of released statistics.

The idea of OBA is that, at every time $t$, we make a decision of either
approximation or publication by considering the benefit in a long-term period of
$w$ timestamps backwards, rather than only the current timestamp.
In more detail, as we already know the dissimilarities $\hdis_{t-w+1}, \ldots,
\hdis_t$, an optimal privacy budget allocation strategy in a time window of size
$w$ should choose publication on those timestamps with large dissimilarities
(say, top $k$ largest dissimilarities), as choosing approximation on these $k$
timestamps is only likely to introduce larger errors in the released statistics.
Hence, a good strategy should spend the remaining $\epsilon/2$ privacy budget on
these $k$ timestamps.
Once a timestamp is selected for publication, its approximation error $\hdis_t$ is replaced by a publication error. The goal then is to minimize the total publication error (the sum of variances) across these $k$ selected timestamps under a fixed total budget of $\epsilon/2$. Since the magnitude of the original $\hdis_t$ no longer affects that timestamp’s error and $\var(\cdot)$ is convex in the budget, the total publication error is minimized by dividing the budget equally among all the $k$ timestamps.
The remaining $w-k$ timestamps with small dissimilarities simply choose
approximation.
Finally, if the current time $t$ belongs to these $k$ timestamps, then the server
chooses publication with privacy budget $\epsilon_{t,2}=\epsilon/(2k)$;
otherwise, the server chooses approximation.
OBA is illustrated in \cref{fig:OBA}.

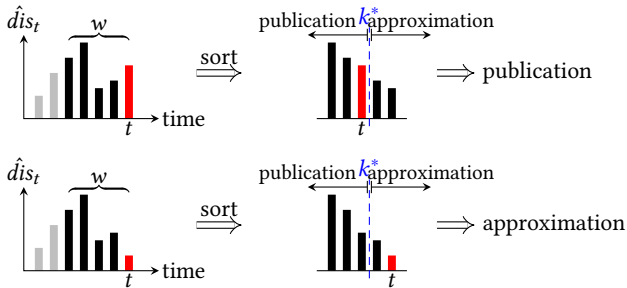
\begin{figure}[htp]
  \centering
  \input{samples/sigmod2026/tz_OBA}
  \caption{Illustration of OBA}
  \label{fig:OBA}
  \Description{The figure demonstrates the Optimal Budget Allocation (OBA) strategy using two scenarios. 
The top row illustrates a "publication" case: the estimated dissimilarity at the current timestamp $t$ (highlighted in red) is high relative to the sliding window w. When the dissimilarities in the window are sorted in descending order, the value at $t$ falls within the top $k^*$ threshold, triggering a data publication. 
The bottom row illustrates an "approximation" case: the dissimilarity at $t$ is low. After sorting, the red bar falls outside the top $k^*$ range into the approximation region, meaning the current data will be approximated using historical values to save privacy budget.}
\end{figure}

The last problem is how to find an optimal $k$.
Let $\hdis_{(1)},\dots,\hdis_{(w)}$ denote the $w$ dissimilarities
$\hdis_{t-w+1}, \dots, \hdis_t$ sorted in descending order.
The \emph{cumulative error} of the released statistics in a time window of size
$w$ consists of two parts: (i) cumulative publication error at these $k$
timestamps due to running $k$ FOs each using a privacy budget
$\frac{\epsilon}{2k}$, thus incurring an error
$k\cdot\var(\frac{\epsilon}{2k})$, and (ii) the cumulative approximation error
on the remaining $w-k$ timestamps, incurring another error
$\sum_{i=k+1}^w\hdis_{(i)}$.
Furthermore, we should allow $k=0$, and in this case, the server always chooses
approximation in the current window, thus incurring an error
$\sum_{i=1}^w\hdis_{(i)}$.

Let $E_t(k)$ denote the cumulative error of the released statistics in the most
recent window of size $w$ at time $t$, expressing $k$ as a parameter.
Then, we have
\[
  E_t(k) =
  \begin{cases}
    \sum_{i=1}^w \hdis_{(i)}, & k=0, \\
    k \cdot \var(\frac{\epsilon}{2k}) + \sum_{i=k+1}^w\hdis_{(i)}, & 0<k\le w.
  \end{cases}
\]
Our goal is to find an optimal $k^*$ by minimizing $E_t(k)$, i.e.,
\begin{equation}\label{eq:optimization}
  k^* \in \arg\min_{0\le k\le w} E_t(k).
\end{equation}
As $k$ is in a finite set $\{0,\ldots,w\}$, a simple method to solve
Problem~\eqref{eq:optimization} is to enumerate each $k$ and choose one
minimizing $E_t(k)$.
The pseudo-code of OBA is given in \cref{alg:oba}.

\input{samples/sigmod2026/alg_opt_budget}

It is noteworthy to mention that the ultimate goal of OBA is to decide whether
we choose approximation or publication at time $t$, and if it is publication,
how much privacy budget to spend.
In Line~\ref{ln:t}, if current timestamp $t$ belongs to the selected $k^*$
timestamps, then we choose publication using privacy budget $\epsilon_{t,2}$;
otherwise we choose approximation with no privacy budget (Line~\ref{ln:zero}).
In addition, $\epsilon_{t,2}$ should be upper bounded by the remaining privacy
budget $\epsilon_t^\mathit{rm}$ available in current window
(Line~\ref{ln:epsilon}).

\paragraph{Improving Efficiency}
A straightforward implementation of OBA requires sorting dissimilarities and enumerating all possible $k$, resulting in a per-timestamp time complexity of $O(w\log w + w^2) =
O(w^2)$.
When $w$ is large, such brute-force enumeration becomes computationally
expensive in streaming settings.
We present several tricks to improve its computational efficiency.
\begin{itemize}
\item{Incremental Maintenance.} 
 Instead of sorting dissimilarities from scratch (Line~\ref{ln:sort}) at each timestamp, we maintain the sorted sequence $\{\hdis_{(1)}, \dots, \hdis_{(w)}\}$ incrementally.
 By using a priority queue, the insertion of the new $\hdis_t$ and the deletion of the expired $\hdis_{t-w}$ can both be performed in $O(\log w)$ time.
\item{Recursive Derivation.} The computation of the cumulative error $E_t(k)$ (Line~\ref{ln:sum2}) can be optimized by reusing intermediate results.
The publication error term $k\cdot\var(\frac{\epsilon}{2k})$ is time
  invariant which depends only on $k$, thus it can be computed offline and stored.
  The publication error term can be computed incrementally while sorting (and similarly for the sum in Line~\ref{ln:sum1}).
  So calculating $E_t$ actually has time complexity $O(w)$.
\item{Early Termination.}
    The function $E_t(k)$ typically exhibits a unimodal structure (initially decreasing and then increasing) or monotonicity.
    Specifically, $E_t(k)$ decreases when the benefit of removing a large approximation error $\hdis_{(k+1)}$ outweighs the cost of added noise, and begins to increase when the remaining dissimilarities are small.
    This structure enables an early-stop strategy: we iterate $k$ starting from $0$ and terminate immediately once $E_t(k+1) > E_t(k)$. Let $k_{\mathit{stop}}$ denote the number of values evaluated before termination.
    Thus, we only need to evaluate $k_{\mathit{stop}}$ candidates instead of all $w$ options.
\end{itemize}
By combining these techniques, the overall time complexity is reduced to $O(\log w + k_{\mathit{stop}})$, ensuring scalability even with large window sizes.

%% file: samples/sigmod2026/tz_OBA.tex
\begin{tikzpicture}[ycomb,font=\small,>=stealth, inner sep=1pt,
  bar/.style={line width=3pt},
  mybrace/.style={decorate,thick,decoration={calligraphic brace,raise=1pt}},
  txt/.style={font=\footnotesize},
  arr/.style={-Implies,double,double distance=2pt},
  ]

  \draw[bar] plot coordinates{(0,.8) (.2,1) (.4,.4) (.6,.5)};
  \draw[bar,color=gray!50] plot coordinates{(-.4,.3) (-.2,.6)};
  \draw[bar,color=red] plot coordinates{(.8,.7)};

  \draw[<->] (-.6,1.1) node[above] {$\hdis_t$} -- (-.6,0) -- (1.2,0) node[right] {time};
  \node[anchor=north] at (.8,0) {$t$};

  \draw[mybrace] (0,1) -- (.8,1) node[midway,above=3pt] {$w$};

  \draw[arr] (1.7,.6) -- ++(.6,0) node[midway,above=2pt] {sort};

  \begin{scope}[xshift=3.5cm]
    \draw[bar] plot coordinates{(0,1) (.2,.8) (.6,.5) (.8,.4)};
    \draw[bar,color=red] plot coordinates{(.4,.7)};
    \node[anchor=north] at (.4,0) {$t$};

    \draw (-.2,0) -- (1,0);

    \draw[densely dashed,blue] (.5,-.1) -- ++(0,1.3) node[above] {$k^*$};

    \draw[Bar->] (.48,1.1) -- ++(-.8,0) node[above,txt] {publication};
    \draw[Bar->] (.52,1.1) -- ++(.8,0) node[above,txt] {approximation};

    \draw[arr] (1.4,.6) -- ++(.5,0) node[right=2pt] {publication};
  \end{scope}

  \begin{scope}[yshift=-2cm]
    \draw[bar] plot coordinates{(0,.8) (.2,1) (.4,.4) (.6,.5)};
    \draw[bar,color=gray!50] plot coordinates{(-.4,.3) (-.2,.6)};
    \draw[bar,color=red] plot coordinates{(.8,.2)};

    \draw[<->] (-.6,1.1) node[above] {$\hdis_t$} -- (-.6,0) -- (1.2,0) node[right] {time};

    \draw[mybrace] (0,1) -- (.8,1) node[midway,above=3pt] {$w$};

    \node[anchor=north] at (.8,0) {$t$};

    \draw[arr] (1.7,.6) -- ++(.6,0) node[midway,above=2pt] {sort};
  \end{scope}

  \begin{scope}[xshift=3.5cm,yshift=-2cm]
    \draw[bar] plot coordinates{(0,1) (.2,.8) (.4,.5) (.6,.4)};
    \draw[bar,color=red] plot coordinates{(.8,.2)};
    \node[anchor=north] at (.8,0) {$t$};

    \draw (-.2,0) -- (1,0);

    \draw[densely dashed,blue] (.5,-.1) -- ++(0,1.3) node[above] {$k^*$};

    \draw[Bar->] (.48,1.1) -- ++(-.8,0) node[above,txt] {publication};
    \draw[Bar->] (.52,1.1) -- ++(.8,0) node[above,txt] {approximation};

    \draw[arr] (1.4,.6) -- ++(.5,0) node[right=2pt] {approximation};
  \end{scope}

\end{tikzpicture}

%% file: samples/sigmod2026/alg_opt_budget.tex
\begin{algorithm}[htp]
  \caption{\textbf{O}ptimal Privacy \textbf{B}udget \textbf{A}llocation
    (OBA)\label{alg:oba}}
  \KwIn{Dissimilarities $\hdis_{t-w+1}, \ldots, \hdis_t$, privacy budget
    $\epsilon$, remaining privacy budget for current window
    $\epsilon_t^\mathit{rm}$} \KwOut{Optimal privacy budget $\epsilon_{t,2}$ for
    timestamp $t$}

  $[\hdis_{(1)},\ldots,\hdis_{(w)}]\gets
  \text{Sort}([\hdis_{t-w+1},\ldots,\hdis_t])$\tcp*{descending}\label{ln:sort}
  $[t_{(1)},\ldots,t_{(w)}]\gets$ timestamps of $[\hdis_{(1)},\ldots,\hdis_{(w)}]$\;
  $E_t[0]\gets\sum_{i=1}^w\hdis_{(i)}$\;\label{ln:sum1}
  \lFor{$k\gets 1$ \KwTo $w$}{%
    $E_t[k]\gets k\cdot\var(\frac{\epsilon}{2k})
    + \sum_{i=k+1}^w\hdis_{(i)}$}\label{ln:sum2}
  $k^*\gets\arg\min_k E_t[k]$\;
  \uIf(\tcp*[f]{if publication}){$t\in\{t_{(1)},\ldots,t_{(k^*)}\}$}{\label{ln:t}
    $\epsilon_{t,2}\gets\min(\epsilon_t^{\mathit{rm}},
    \frac{\epsilon}{2k^*})$\tcp*{publication using budget
      $\epsilon_{t,2}$}\label{ln:epsilon}
  }
  \lElse(\tcp*[f]{approximation using previous time step}){%
    $\epsilon_{t,2}\gets 0$}\label{ln:zero}
  \Return $\epsilon_{t,2}$\;
\end{algorithm}

%% file: samples/sigmod2026/ss_tree_release.tex
\subsection{Private Adaptive Tree Publication}
\label{ss:tree_publication}

In the previous mechanism, if the output privacy budget $\epsilon_{t,2}>0$, then
the server will choose publication with error using privacy budget
$\epsilon_{t,2}$ to release current stream statistics, which are represented as
a private binary tree, denoted by $\OT_t$.
We now describe how to build tree $\OT_t$ using privacy budget $\epsilon_{t,2}$.

A straightforward way to build tree $\OT_t$ is to use the method we introduced
in \cref{ss:tree}.
That is, we build $\OT_t$ by invoking an FO at each layer with privacy budget $\epsilon_{t,2}$.
The issue of this approach is that, user data $D_t$ in practice may be not
evenly distributed.
There are few ``hot'' values (e.g., few popular places many people visit), and
many ``cold'' values (e.g., many places people rarely visit) in $D_t$.
In this case, nodes in $\OT_t$ representing cold values will have very small
property values (i.e., frequencies), and after adding noise by FO, their
estimation accuracy is likely to diminish, resulting in little to no utility.

To address these issues, a data-adaptive hierarchical structure is commonly employed for static data. Such a structure is typically constructed level by level: at each level, a central server interacts with users to obtain interval frequency counts and decides whether to further partition those intervals. Although this approach incurs latency that grows with tree height and is impractical in streaming scenarios, it is acceptable in static scenarios, where the private tree needs to be built only once. However, efficiently addressing this challenge in streaming data, where low latency is required, remains an open problem.

A key novelty of MTSP-LDP lies in its effective use of the intermediate tree $\hDT_t$ from $\CM_1^t$, which has been largely overlooked in prior methods.
We propose a \emph{Data-\textbf{A}daptive Private Binary
  \textbf{T}ree \textbf{C}onstruction (ATC)} method that builds a private binary
tree $\hAT_t$ able to better capture the statistics of user data $D_t$ than the
straightforward method.
Furthermore, we propose a \emph{grouping and smoothing} strategy to refine trees
$\{\hAT_{t'}\}_{t'\leq t}$, which can further improve the estimation accuracy of
the final released tree $\OT_t$.

\subsubsection{Data-Adaptive Private Binary Tree Construction (ATC)}

The idea of ATC is to leverage the tree $\hDT_t$ as an auxiliary tree to build a
data-adaptive private binary tree $\hAT_t$ to better capture the statistics of
user data $D_t$ than building the tree from scratch.
Recall that $\hDT_t$ is also built from $D_t$ but using privacy budget
$\epsilon_{t,1} = \epsilon/(2w)$ at time $t$.
Hence, the tree $\hDT_t$ also contains information about $D_t$.
We propose to use the node properties stored in $\hDT_t$ to decide whether we
need to prune the tree, i.e., remove branches representing cold values.
This approach allows the entire adaptive tree structure to be constructed at once, without consuming additional privacy budget or requiring extra rounds of user interaction.
Then, FO is only applied to the pruned tree to estimate node properties.
Finally, node properties in removed branches are directly inferred from their
parents.

\begin{figure}[htp]
  \centering
  \input{samples/sigmod2026/tz_tree_release}
  \caption{Illustration of ATC}
  \label{fig:ATC}
  \Description{The figure illustrates the Adaptive Tree Construction (ATC) process in three stages. 
On the left, the initial tree shows a parent node 'a' with its child nodes. 
The middle stage shows the result of a pruning operation, where the children of node 'a' have been removed. 
The final stage on the right shows the tree after the FO. Here, the previously pruned child nodes are restored as nodes 'b' and 'c' (highlighted in blue). This indicates that their values are estimated by uniformly dividing the frequency of their parent node 'a', rather than being queried directly.}
\end{figure}
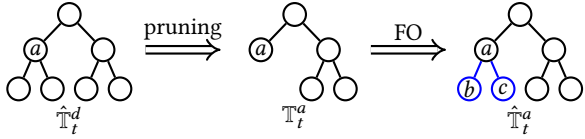

We use the example in \cref{fig:ATC} to further explain the above idea.
Given $\hDT_t$ from $\CM_1^t$, we check the node property from top to bottom,
layer by layer.
If some node $a$'s property $\hDT_t[a]<\vartheta_1$, where $\vartheta_1$ denotes
a threshold frequency, then we prune the subtree rooted at node $a$ while only
keeping node $a$.
This results a pruned tree $\AT_t$.
Then, we run two FOs to estimate the bottom two layers' node properties
(cf.~\cref{ss:tree}), and obtain a tree $\hAT_t$.
For pruned nodes, their properties are directly inferred from their parents,
e.g., $\hAT_t[b]=\hAT_t[c]=\hAT_t[a]/2$.
\input{samples/sigmod2026/alg_adaptive_tree}
The pseudo-code of ATC is given in \cref{alg:ATC}.
We first create a pruned tree $\AT_t$ from the tree $\hDT_t$
(Lines~\ref{ln:prune_start} to~\ref{ln:prune_end}).
Then we invoke an FO to estimate each node's property and obtain the tree $\hAT_t$
(Line~\ref{ln:FO}).
Finally, pruned nodes' properties are directly inferred from their parents
(Lines~\ref{ln:infer_start} to~\ref{ln:infer_end}).
It is worth noting that the resulting tree $\hAT_t$ is still a perfect
binary tree.

\subsubsection{Grouping and Smoothing}
\label{ss:group_and_smooth}
The data-adaptive private binary tree $\hAT_t$ captures the instantaneous
statistics of the stream at time $t$, which may change a lot over time if the
stream is highly variable.
Instead of releasing $\hAT_t$ as the stream statistics, we propose to release
steady statistics, which can capture the underlying trend of the stream.
Queries on the steady statistics will be more accurate and meaningful than on
the instantaneous statistics.

To this end, inspired by Pegasus~\cite{Chen2017}, we introduce a \emph{grouping
  and smoothing} post-processing module.
Different from Pegasus which is designed for CDP and user data is available to
the server, we design grouping and smoothing module in the $w$-event LDP setting
without directly accessing user data.
The grouping and smoothing module processes trees $\{\hAT_{t'}\}_{t'\leq t}$,
and outputs steady statistics $\{\OT_{t'}\}_{t'\leq t}$, which are also private
binary trees.

\paragraph{Grouping}
For trees $\{\hAT_{t'}\}_{t'\leq t}$, we consider a same node $a$ at different
time in these trees.
We want to group the most recent similar node properties of node $a$ into a
group, and use the aggregate statistics of node property in this group as node
$a$'s property in the released tree $\OT_t$.
Specifically, let $T_a=\{t-l,\ldots,t-1\}$ denote a set of $l$ timestamps at
which node $a$'s properties are similar and belong to a group at time $t-1$.
At time $t$, we need to determine whether $t$ belongs to group $T_a$ or not.
Recall that $\hAT_t[a]$ is the estimated node property of $a$ at time $t$, and
we let $\AT_t[a]$ denote its true property value.
To measure the deviation of $\AT_t[a]$ to the group, we define the \emph{squared
  deviation} of node $a$ by
\[
  \sigma_a^2\triangleq (\AT_t[a] - \frac{1}{l}\sum_{i=1}^l\AT_{t-i}[a])^2,
\]
and if $\sigma_a^2\leq\vartheta_2$ for some threshold $\vartheta_2$, we add $t$ to
$T_a$; otherwise node $a$'s group becomes to $T_a=\{t\}$ at time $t$.
The challenge of the grouping operation is that the true node properties
$\{\AT_{t'}[a]\}_{t'\leq t}$ are unknown in the LDP setting, and we only know
their estimates $\{\hAT_{t'}[a]\}_{t'\leq t}$.
We provide an unbiased estimator of $\sigma_a^2$.

\begin{theorem}\label{thm:sigma}
  Estimator $\hat{\sigma}_a^2$ is an unbiased estimate of $\sigma_a^2$, i.e.,
  \[
    \hat{\sigma}_a^2
    = \big(\hAT_t[a] -\frac{1}{l}\sum_{i=1}^l\hAT_{t-i}[a]\big)^2
    -\frac{l+1}{l}\var(\epsilon_{t,2}).
  \]
\end{theorem}


\begin{proof}
First, compute the expectation of $\hat{\sigma}_a^2 + \frac{l+1}{l} \var(\epsilon_{t,2})$:
{\allowdisplaybreaks
\begin{align*}
  &\quad \E\big(\hat{\sigma}_a^2+\frac{l+1}{l}\var(\epsilon_{t,2})\big) \\
  &= \E\big[ \big(\hAT_t[a] -\frac{1}{l}\sum_{i=1}^l\hAT_{t-i}[a]\big)^2 \big] \\
  &= \E\big[ \hAT_t[a]^2 -\frac{2}{l}\hAT_t[a]\sum_{i=1}^l\hAT_{t-i}[a]
    +(\frac{1}{l}\sum_{i=1}^l\hAT_{t-i}[a])^2 \big] \\
  &= \E(\hAT_t[a]^2) -\frac{2}{l}\AT_t[a]\sum_{i=1}^l\AT_{t-i}[a]
    +\E\big[ (\frac{1}{l}\sum_{i=1}^l\hAT_{t-i}[a])^2 \big] \\
  &= \var(\hAT_t[a]) + \AT_t[a]^2 -\frac{2}{l}\AT_t[a]\sum_{i=1}^l\AT_{t-i}[a] \\
  &\quad +\var\big( \frac{1}{l}\sum_{i=1}^l\hAT_{t-i}[a] \big)
    +(\frac{1}{l}\sum_{i=1}^l\AT_{t-i}[a])^2 \\
  &= (\AT_t[a]-\frac{1}{l}\sum_{i=1}^l\AT_{t-i}[a])^2
    + \var(\hAT_t[a]) +\var\big( \frac{1}{l}\sum_{i=1}^l\hAT_{t-i}[a] \big) \\
  &= \sigma_a^2
    + \var(\hAT_t[a]) +\var\big( \frac{1}{l}\sum_{i=1}^l\hAT_{t-i}[a] \big)
\end{align*}
}
Notice that
\[
  \var(\hAT_t[a]) = \var(\epsilon_{t,2})
\]
and
\[
  \var\big( \frac{1}{l}\sum_{i=1}^l\hAT_{t-i}[a] \big)
  = \frac{1}{l^2}\sum_{i=1}^l\var(\hAT_{t-i}[a])
  = \frac{1}{l}\var(\epsilon_{t,2})
\]
Therefore, we conclude that
\[
  \E(\hat{\sigma}_a^2) = \sigma_a^2.
\]
This completes the proof.
\end{proof}

\paragraph{Smoothing}
At each time $t$, we now know every node $a$ belongs to a most recent group
$T_a$ with similar node properties at different time.
To release the statistics with regard to node $a$, we propose to compute the
aggregate statistics of node properties in the group, e.g., mean, median, etc.
Hence, we define
\[
  \OT_t[a] \triangleq \text{Aggregate}(\{\hAT_{t'}[a]\colon t'\in T_a\})
\]
and release tree $\OT_t$ as the stream statistics at time $t$.
Querying on $\OT_t$ will provide more accurate and meaningful results than on
$\hAT_t$.

\paragraph{Remarks}
It is noteworthy to mention that grouping and smoothing are post-processing
operations, both of which operate on perturbed statistics without consuming
additional privacy budgets.

%% file: samples/sigmod2026/tz_tree_release.tex
\begin{tikzpicture}[thick,font=\small,
  nd0/.style={circle,minimum size=3mm,inner sep=0pt},
  nd1/.style={nd0,draw},
  nd2/.style={nd0,draw,blue},
  nd/.style={nd1},
  arr/.style={-Implies,double,double distance=2pt},
  ]

  \coordinate (O1);
  \foreach \i in {0,...,3}{
    \node[nd,right=.45*\i of O1] (a2\i) {};
  }
  \path (a20)--(a21) node[midway,above=.35,nd] (a10) {$a$};
  \path (a22)--(a23) node[midway,above=.35,nd] (a11) {};
  \path (a10)--(a11) node[midway,above=.3,nd] (a0) {};

  \draw (a10) -- (a0) -- (a11);
  \draw (a20) -- (a10) -- (a21);
  \draw (a22) -- (a11) -- (a23);

  \node[below=.9 of a0] {$\hDT_t$};

  \coordinate[right=.4 of a11] (A);
  \draw[arr] (A) -- ++(1,0) node[midway,above] {pruning};

  \coordinate[right=3 of O1] (O2);
  \foreach \i/\j in {0/0,1/0,2/1,3/1}{
    \node[nd\j,right=.45*\i of O2] (b2\i) {};
  }
  \path (b20)--(b21) node[midway,above=.35,nd] (b10) {$a$};
  \path (b22)--(b23) node[midway,above=.35,nd] (b11) {};
  \path (b10)--(b11) node[midway,above=.3,nd] (b0) {};

  \draw (b10) -- (b0) -- (b11);
  \draw (b22) -- (b11) -- (b23);

  \node[below=.9 of b0] {$\AT_t$};

  \coordinate[right=.4 of b11] (A);
  \draw[arr] (A) -- ++(1,0) node[midway,above] {FO};

  \coordinate[right=3 of O2] (O3);
  \foreach \i/\j in {0/2,1/2,2/1,3/1}{
    \node[nd\j,right=.45*\i of O3] (c2\i) {};
  }
  \node[nd0] at (c20) {$b$};
  \node[nd0] at (c21) {$c$};

  \path (c20)--(c21) node[midway,above=.35,nd] (c10) {$a$};
  \path (c22)--(c23) node[midway,above=.35,nd] (c11) {};
  \path (c10)--(c11) node[midway,above=.3,nd] (c0) {};

  \draw (c10) -- (c0) -- (c11);
  \draw[blue] (c20) -- (c10) -- (c21);
  \draw (c22) -- (c11) -- (c23);

  \node[below=.9 of c0] {$\hAT_t$};
\end{tikzpicture}

%% file: samples/sigmod2026/alg_adaptive_tree.tex
\begin{algorithm}[htp]
  \caption{Data-\textbf{A}daptive Private Binary \textbf{T}ree
    \textbf{C}onstruction (ATC)}
  \label{alg:ATC}
  \KwIn{Private binary tree $\hDT_t$, threshold $\vartheta_1$,
    privacy budget $\epsilon_{t,2}$}
  \KwOut{Private binary tree $\hAT_t$}

  \tcp{Build tree $\AT_t$ by pruning tree $\hDT_t$}
  Initialize an empty tree $\AT_t$ with only root node $r$\;\label{ln:prune_start}
  $\mathit{current\_layer}\gets\{r\}$, $level \gets 0$,
  $h\gets$ height of tree $\hDT_t$\;
  \While{$\mathit{current\_layer} \neq \emptyset$ and $level < h$}{
    $\mathit{next\_layer}\gets\emptyset$\;
    \ForEach{node $ a\in \mathit{current\_layer}$}{
      \If{$\hDT_t[a]\geq\vartheta_1$}{
        Equally split node $a$ into two child nodes $c_1, c_2$\;
        $\mathit{next\_layer}\gets\mathit{next\_layer}\cup\{c_1,c_2\}$\;
      }
    }
    $level \gets level + 1$\;
    $\mathit{current\_layer}\gets\mathit{next\_layer}$\;\label{ln:prune_end}
  }

  \tcp{Estimate each node's property using an FO (cf.~\cref{ss:tree})}
  $\hAT_t\gets\FO(\AT_t, \epsilon_{t,2})$\;\label{ln:FO}

  \tcp{Infer nodes' properties in the pruned tree}
  \ForEach{node $a\in\hAT_t$}{\label{ln:infer_start}
    \If{node $a$ has no child and its depth $<h$}{
      Crete two child nodes $c_1, c_2$ for node $a$\;
      $\hAT_t[c_1]\gets\hAT_t[a]/2$,
      $\hAT_t[c_2]\gets\hAT_t[a]/2$\;\label{ln:infer_end}
    }
  }
  \Return $\hAT_t$\;
\end{algorithm}

%% file: samples/sigmod2026/ss_query_processing.tex
\subsection{Budget-Free Multi-Task Streaming Query}

Based on the released stream statistics $\{\OT_{t'}\}_{t'\leq t}$, we can process
multiple streaming query tasks defined in \cref{ss:queries}.

\subsubsection{Counting Query}

To answer a counting query $Q_t^c(v,\Delta)$ at time $t$, we traverse each tree
$\OT_{t'}$ such that $t - t' < \Delta$, and retrieve the property of the leaf
node $a_v$ representing private value $v$, and hence
\[
  Q_t^c(v,\Delta) = \sum_{t'=t-\Delta+1}^t \OT_{t'}[a_v]n_{t'}.
\]

\subsubsection{Range Query}

To answer a range query $Q_t^r([v_1,v_2], \Delta)$ at time $t$, we find a
minimum cover in the released tree such that they jointly cover the query range
$[v_1,v_2]$ (cf.~\cref{fig:tree}).
Denote the minimum cover by $c(v_1,v_2)$, then
\[
  Q_t^r([v_1, v_2], \Delta) =
  \sum_{t'=t-\Delta+1}^t\sum_{a\in c(v_1,v_2)} \OT_{t'}[a]n_{t'}.
\]

\subsubsection{Event Monitoring}

Because an event monitoring query $Q_t^e(\alpha,\beta)$ actually consists of
several counting or range queries, it can be efficiently
answered as above.

%% file: samples/sigmod2026/analysis.tex
\section{Privacy Analysis and Parameter Selection}

\input{samples/sigmod2026/ss_privacy_analysis}

\input{samples/sigmod2026/ss_parameter_selection}

%% file: samples/sigmod2026/ss_privacy_analysis.tex
\subsection{Privacy Analysis}

We have the following privacy guarantee for MTSP-LDP.

\begin{theorem}
  MTSP-LDP satisfies $w$-event $\epsilon$-LDP for each user.
\end{theorem}

\begin{proof}
  MTSP-LDP ensures that the total privacy budget consumed within any sliding
  window of length $w$ does not exceed $\epsilon$.
  Specifically, at each timestamp $t$, the budget $\epsilon$ is partitioned into
  two parts, i.e., $\epsilon/2$ is evenly divided among the $w$ mechanisms
  $\CM_1^{t'}$ where $t'\in [t-w+1, t]$, and the remaining $\epsilon/2$ is
  dynamically distributed among the corresponding $w$ mechanisms $\CM_2^{t'}$.
  As explained in \cref{ss:optimal_budget}, the allocation strategy guarantees
  that the cumulative budget consumed by all $\CM_2^{t'}$ over any sliding
  window of length $w$ does not exceed $\epsilon/2$.
  Mechanisms $\CM_3^t$ and $\CM_4^t$ only process perturbed data and do not
  access user data, i.e., they are post-processing operations, thus incurring no
  additional privacy cost.
  Therefore, the total privacy budget spent within any sliding window is bounded
  by $\epsilon$, and MTSP-LDP satisfies $w$-event $\epsilon$-LDP.
\end{proof}

%% file: samples/sigmod2026/ss_parameter_selection.tex
\subsection{Selection of Thresholds $\vartheta_1$ and $\vartheta_2$}

Two important parameters in MTSP-LDP are thresholds $\vartheta_1$ and $\vartheta_2$
used in \cref{ss:tree_publication}.
We discuss how to choose these thresholds to maximize the overall utility of our
framework.

\input{samples/sigmod2026/sss_theta_1}

\input{samples/sigmod2026/sss_theta_2}

%% file: samples/sigmod2026/sss_theta_1.tex
\subsubsection{Selection of Threshold $\vartheta_1$}
Threshold $\vartheta_1$ is used for pruning tree $\hDT_t$ in \cref{alg:ATC} to
obtain a data-adaptive tree $\hAT_t$ after invoking FOs for each layer.
We find that invoking FOs on the pruned tree may reduce estimation error if
threshold $\vartheta_1$ is carefully chosen.
To understand the reason, without loss of generality, let us consider a tree
without pruning and a tree pruned at node $a$, respectively (see
\cref{fig:branch_of_tree}).
If we spend the same amount of privacy budget $\epsilon_{t,2}$ and use the same
number of users to estimate the node frequencies in each layer of these two
trees, then the estimate of each node frequency will have the same variance
$\var(\epsilon_{t,2})$.
While for the pruned tree, properties of those pruned nodes are approximated by
recursively halving their parent nodes' properties.
Therefore, estimation errors are different only for the two subtrees rooted at
node $a$, excluding node $a$.
Assume the subtree has height $h+1$.

For the left tree, the total estimation error of the subtree rooted at node $a$
(excluding node $a$) is
\[
  \err_1 \triangleq
  \sum_{i=1}^h 2^i\E[(\hat{f_i} - f_i)^2]
  = 2(2^h-1)\var(\epsilon_{t,2}).
\]

For the right tree, the total estimation error of the subtree rooted at node $a$
(excluding node $a$) is
\begin{align*}
  \err_2
  &\triangleq \sum_{i=1}^h 2^i\E[(\hat{f}_i-f_i)^2]
    = \sum_{i=1}^h 2^i\E[(\frac{\hat{f}_0}{2^i}-f_i)^2] \\
  &= \sum_{i=1}^h 2^i\E( \frac{\hat{f}_0^2}{2^{2i}} - \frac{\hat{f}_0f_i}{2^{i-1}}
    +f_i^2 )
  = \sum_{i=1}^h 2^i(\frac{\var(\epsilon_{t,2}) + f_0^2}{2^{2i}}
    - \frac{f_0f_i}{2^{i-1}} + f_i^2) \\
  &\leq \sum_{i=1}^h 2^i(\frac{\var(\epsilon_{t,2})+\vartheta_1^2}{2^{2i}}+\vartheta_1^2)
    = \sum_{i=1}^h(\frac{\var(\epsilon_{t,2})+\vartheta_1^2}{2^i} + 2^i\vartheta_1^2) \\
  &\leq \var(\epsilon_{t,2}) + \vartheta_1^2 + 2(2^h-1)\vartheta_1^2 \\
  &= \var(\epsilon_{t,2}) + (2^{h+1}-1)\vartheta_1^2.
\end{align*}
If we require the estimation error for the right tree is no larger than the left
tree, i.e., $\err_2\le \err_1$, then we obtain
\begin{equation}\label{eq:theta1}
  \vartheta_1
  \leq \sqrt{\big(\frac{2^{h+1}-3}{2^{h+1}-1}\big)\cdot\var(\epsilon_{t,2})}
  \leq \sqrt{\var(\epsilon_{t,2})}.
\end{equation}
Therefore, if threshold $\vartheta_1$ is chosen according to
Condition~\eqref{eq:theta1}, the pruned tree will have smaller overall
estimation error than the full binary tree, thus demonstrating the usefulness of
tree pruning operation in \cref{alg:ATC}.
\begin{figure}[htp]
  \centering
  \input{samples/sigmod2026/tz_branch}
  \caption{Example of a same tree without pruning (left) and pruning at node $a$
  (right)}
  \label{fig:branch_of_tree}
  \Description{This figure illustrates a comparison between a tree structure before and after pruning. The tree on the left shows the original structure, where all branches are intact.}
\end{figure}
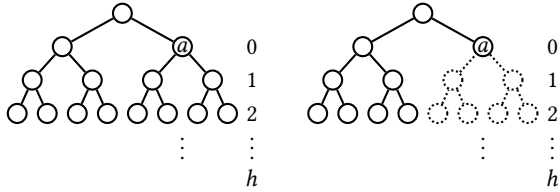

%% file: samples/sigmod2026/tz_branch.tex
\begin{tikzpicture}[
  thick, font=\small, inner sep=0pt,
  nd/.style={circle,minimum size=2.5mm},
  nd0/.style={nd,draw,densely dotted},
  nd1/.style={nd,draw},
  line0/.style={densely dotted},
  line1/.style={},
  ]

  \coordinate (O1);
  \foreach \i in {0,...,7}{
    \node[nd1,right=.4*\i of O1] (a3\i) {};
  }

  \foreach \i in {0,...,3}{
    \path (a3\the\numexpr\i*2\relax) -- (a3\the\numexpr\i*2+1\relax)
    node[midway,above=.3,nd1] (a2\i) {};

    \draw (a3\the\numexpr\i*2\relax) -- (a2\i) -- (a3\the\numexpr\i*2+1\relax);
  }

  \foreach \i in {0,1}{
    \path (a2\the\numexpr\i*2\relax) -- (a2\the\numexpr\i*2+1\relax)
    node[midway,above=.3,nd1] (a1\i) {};

    \draw (a2\the\numexpr\i*2\relax) -- (a1\i) -- (a2\the\numexpr\i*2+1\relax);
  }
  \path (a10) -- (a11) node[midway,above=.3,nd1] (a0) {};
  \draw (a10) -- (a0) -- (a11);

  \node at (a11) {$a$};
  \node[right=.7 of a11] (l0) {$0$};
  \node at (l0 |- a20) (l1) {1};
  \node at (l0 |- a30) (l2) {$2$};
  \node[below=0 of l2] (ld) {$\vdots$};
  \node[below=.1 of ld] (lh) {$h$};
  \node at (a11 |- ld) {$\vdots$};

  \coordinate[right=4 of O1] (O2);
  \foreach \i/\j in {0/1,1/1,2/1,3/1,4/0,5/0,6/0,7/0}{
    \node[nd\j,right=.4*\i of O2] (b3\i) {};
  }

  \foreach \i/\j in {0/1,1/1,2/0,3/0}{
    \path (b3\the\numexpr\i*2\relax) -- (b3\the\numexpr\i*2+1\relax)
    node[midway,above=.3,nd\j] (b2\i) {};

    \draw[line\j] (b3\the\numexpr\i*2\relax) -- (b2\i) --
    (b3\the\numexpr\i*2+1\relax);
  }

  \foreach \i/\j in {0/1,1/0}{
    \path (b2\the\numexpr\i*2\relax) -- (b2\the\numexpr\i*2+1\relax)
    node[midway,above=.3,nd1] (b1\i) {};

    \draw[line\j] (b2\the\numexpr\i*2\relax) -- (b1\i) --
    (b2\the\numexpr\i*2+1\relax);
  }

  \path (b10) -- (b11) node[midway,above=.3,nd1] (b0) {};
  \draw (b10) -- (b0) -- (b11);


  \node at (b11) {$a$};
  \node[right=.7 of b11] (l0) {$0$};
  \node at (l0 |- b20) (l1) {1};
  \node at (l0 |- a30) (l2) {$2$};
  \node[below=0 of l2] (ld) {$\vdots$};
  \node[below=.1 of ld] (lh) {$h$};
  \node at (b11 |- ld) {$\vdots$};

\end{tikzpicture}

%% file: samples/sigmod2026/sss_theta_2.tex
\subsubsection{Selection of Threshold $\vartheta_2$}

The threshold $\vartheta_2$ is used in the grouping and smoothing module
(cf.~\cref{ss:group_and_smooth}).
We find that the proper choice of $\vartheta_2$ can also reduce estimation error.

To understand the reason, let us focus on an arbitrary node $a$ in the tree.
Let $T_a=\{t-l,\ldots,t\}$ denote a group of $l+1$ timestamps at which node
$a$'s properties are close.
Let $g_i\triangleq\AT_{t-i}[a]$ denote the true property value and
$\hat{g}_i\triangleq\hAT_{t-i}[a]$ denote its estimate.
Let $\bar{g}_0=\frac{1}{l+1}\sum_{i=0}^{l}\hat{g}_i$ denote the smoothed
property estimate at time $t$ (i.e., assume the aggregation function is the mean).

Without the grouping and smoothing operation, the estimation variance at time $t$
is
\[
  \err_3 \triangleq \E[(\hat{g}_0-g_0)^2] = \var(\epsilon_{t,2}).
\]
In contrast, if we apply the grouping and smoothing operation, the estimation
variance becomes
\begin{align*}
  \err_4
  &\triangleq \E[(\bar{g}_0-g_0)^2]
    = \E\big[(\frac{1}{l+1}\sum_{i=0}^l\hat{g}_i - g_0)^2\big] \\
  &= \E\big[(\frac{1}{l+1}\sum_{i=0}^l\hat{g}_i)^2
    -\frac{2}{l+1}g_0\sum_{i=0}^l\hat{g}_i  + g_0^2\big] \\
  &= \var(\frac{1}{l+1}\sum_{i=0}^l\hat{g}_i)
    + (\frac{1}{l+1}\sum_{i=0}^lg_i - g_0)^2 \\
  &= (\frac{1}{l+1})^2\sum_{i=0}^l\var(\epsilon_{t-i,2})
    + \big[\frac{l}{l+1}(\frac{1}{l}\sum_{i=1}^lg_i - g_0)\big]^2 \\
  &= (\frac{1}{l+1})^2\sum_{i=0}^l\var(\epsilon_{t-i,2})
    + (\frac{l}{l+1})^2\sigma_a^2.
\end{align*}
We require $\err_4\le\err_3$, implying
\[
  \sigma_a^2\leq\frac{1}{l^2}\big[ (l+1)^2\var(\epsilon_{t,2})
  -\sum_{i=0}^l\var(\epsilon_{t-i,2})\big].
\]
Therefore, if we maintain the group for node $a$ as long as the following condition holds, i.e.,
\begin{equation}\label{eq:theta2}
  \vartheta_2
  \leq\frac{1}{l^2}\big[ (l+1)^2\var(\epsilon_{t,2})
  -\sum_{i=0}^l\var(\epsilon_{t-i,2})\big],
\end{equation}
then the grouping and smoothing operation can reduce estimation error in our
framework.

%% file: samples/sigmod2026/experiments.tex
\section{Evaluation} \label{sec:Evaluation}

In this section, we evaluate the performance of MTSP-LDP on real-world datasets.

\input{samples/sigmod2026/ss_dataset}

\input{samples/sigmod2026/ss_settings}

\input{samples/sigmod2026/ss_results}

%% file: samples/sigmod2026/ss_dataset.tex
\subsection{Datasets}

We use four publicly available datasets, and we briefly describe them below.
\begin{itemize}
    \item Cosmetics\footnotemark[1] is a seven-month dataset (Oct.~2019 to Apr.~2020) from a large e-commerce platform. We extract the last viewed item of each user per day, and this forms a value set of size $329$. The stream consists of $1,654,771$ users with $5,112$ timestamps.
    
    \item Taxi\footnotemark[2] is a collection of NYC yellow taxi trips from Jan.~to Sep.~2024.
    Fare values are aggregated on a daily basis, and values above the $99.9$-th percentile are removed. This yields a value set of size $150$.
    The stream consists of $139,713$ users with $275$ timestamps.
    
    \item Loan\footnotemark[3] is a collection of Lending Club loan records from
    2007 to 2018.
    We aggregate loan amounts by issuance month and round values downward.
    This results in a value set of size $1,719$.
    The stream consists of $61,992$ users with $138$ timestamps.
    
    \item Foursquare\footnotemark[4] is a collection of check-ins collected from
    a location-based social network Foursquare from Apr.~2012 to Sep.~2013.
    The value set is a set of cities of size $415$.
    We count daily city-level check-ins and obtain a stream of $266,909$ users over
    $447$ timestamps.
\end{itemize}
The statistics of these datasets are summarized in \cref{tab:dataset}.

\input{samples/sigmod2026/tab_dataset}

\footnotetext[1]{\url{https://www.kaggle.com/datasets/mkechinov/ecommerce-events-history-in-cosmetics-shop}}
\footnotetext[2]{\url{https://www.nyc.gov/site/tlc/about/tlc-trip-record-data.page}}
\footnotetext[3]{\url{https://www.kaggle.com/datasets/wordsforthewise/lending-club}}
\footnotetext[4]{\url{https://sites.google.com/site/yangdingqi/home/foursquare-dataset}}

%% file: samples/sigmod2026/tab_dataset.tex
\begin{table}[htp]
  \centering
  \caption{Dataset statistics}
  \label{tab:dataset}
  \begin{tblr}{hlines,
    colspec={c|c|c|c},
    cells={font=\small,c,m},
    row{1} = {font=\small\bfseries,c},
    rowsep=0pt,
    }
    \hline
    dataset    & stream length & domain size $d$ & \# of users $n$ \\
    Cosmetics  & $5,112$       & $329$           & $1,654,771$     \\
    Taxi       & $275$         & $150$           & $139,713$       \\
    Loan       & $138$         & $1,719$         & $61,992$        \\
    Foursquare & $447$         & $415$           & $266,909$       \\
   \end{tblr}
\end{table}

%% file: samples/sigmod2026/ss_settings.tex
\subsection{Settings}

\subsubsection{Metrics}

To quantify the performance of the algorithm, we use different metrics tailored
to each query task.
Let $\hat\theta$ denote an estimator of the true value $\theta$, and let
$\hat\theta_i$ denote the estimate of the true value $\theta_i$ in the $i$-th query,
for $i=1,\ldots,n_q$.
\begin{itemize}
    \item {Mean Absolute Error (MAE)}.
MAE quantifies the average magnitude of estimation errors and serves as a
standard indicator of overall accuracy across all query types.
The MAE of an estimator $\hat\theta$ is defined as
\[
  \mathsf{MAE}(\hat\theta)
  \triangleq \E(|\hat\theta - \theta |)
  = \frac{1}{n_q}\sum_{i=1}^{n_q} |\hat\theta_i - \theta_i |.
\]
    \item {Mean Relative Error (MRE)}.
MRE measures the error relative to the true value, offering
a scale-invariant perspective on accuracy.
The MRE of an estimator $\hat\theta$ is defined as
\[
  \mathsf{MRE}(\hat\theta)
  \triangleq \E(\frac{|\hat\theta - \theta |}{\theta})
  = \frac{1}{n_q}\sum_{i=1}^{n_q}\frac{|\hat\theta_i - \theta_i|}{\theta_i}.
\]
    \item {Receiver Operating Characteristic (ROC) Curve}.
For event monitoring, we use the ROC curve to evaluate the trade-off between
true positive and false positive rates across varying decision thresholds.
The Area Under the Curve (AUC) is used as a key metric for evaluating the
model's overall performance.
A higher AUC indicates better event detection performance.
\end{itemize}



\subsubsection{Baselines}

We compare MTSP-LDP with several state-of-the-art $w$-event LDP algorithms,
including LBA, LBD, LSP, and LBU, as introduced in \cref{ss:competitors}.
To ensure fair comparison and reproducibility, all methods, including MTSP-LDP,
were implemented in Python under a unified experimental framework.
All experiments were conducted on a desktop equipped with an Intel Core i7-10700
CPU (2.90 GHz) and 16 GB of RAM.

%% file: samples/sigmod2026/ss_results.tex
\subsection{Results}

\subsubsection{Evaluation for Counting Queries}

\input{samples/sigmod2026/exp_counting_query_figures}

In order to compare MTSP-LDP fairly with other baselines which are mainly
designed for frequency histogram estimation, we let $\hat\theta_{t,v}\triangleq
Q_t^c(v,1)/n_t$, and hence $\{\hat\theta_{t,v}\}_{v\in\Omega}$ is an estimate of
the frequency histogram of the stream at time $t$.
Then we evaluate the MAE and MRE of $\hat\theta_{v,t}$, averaged over time, and show
the results on the four datasets in \cref{fig:histogram}.

We observe that, for a fixed sliding window size $w$, as the privacy budget
$\epsilon$ increases from $0.5$ (i.e., high privacy protection) to $5$ (i.e., low
privacy protection), the estimation error for all algorithms decreases.
Similarly, for a fixed privacy budget $\epsilon$, as the window size $w$
increases from $10$ to $50$, the estimation error increases accordingly.
These results highlight a clear trade-off between privacy and data utility,
where larger privacy budgets or smaller window sizes generally lead to better
utility.

Notably, MTSP-LDP delivers consistently high accuracy across diverse datasets.
As discussed in \cref{sec:method}, this can be attributed to the fact that
existing algorithms do not fully leverage the temporal correlations in the data
stream when allocating privacy budgets.
Although the LBD and LBA methods conserve privacy budgets during periods of
minimal data change, their strategies focus solely on the current timestamp
without considering the data distribution across the entire time window.
This limitation is particularly critical in the $w$-event setting, where
capturing the overall data distribution is crucial for accurate estimates.
Besides, MTSP-LDP demonstrates a particularly significant advantage when the
privacy budget is small, as smaller budgets inherently introduce larger errors,
amplifying the weaknesses of other methods.
This makes MTSP-LDP highly effective in scenarios requiring strict privacy
guarantees.

\subsubsection{Evaluation for Range Queries}

\input{samples/sigmod2026/exp_range_query_figures}

For baseline methods with frequency histogram estimates
$\{\hat{f}_{t,v}\}_{v\in\Omega}$, we can obtain the estimate of a frequency range query by
\[
  \hat{f}_{t,[v_1,v_2]}
  \triangleq \sum_{t'=t-\Delta+1}^t\sum_{v\in[v_1,v_2]} \hat{f}_{t',v}.
\]
For MTSP-LDP, the corresponding estimate is
\[
  \hat\theta_{t,[v_1,v_2]}
  \triangleq \frac{Q_t^r([v_1,v_2],\Delta)}{\sum_{t'=t-\Delta+1}^t n_{t'}}.
\]
We randomly generate $50$ range query tasks with varying value ranges
$[v_1,v_2]$ and time spans $\Delta$.
For each query, the MRE and MAE between the estimates and ground truths are
computed, and the average MRE and MAE across all queries are reported.
The results are shown in \cref{fig:range_query}.
Note that some query tasks span time ranges exceeding the predefined window
size, which means the total privacy budget used by these queries exceeds the set
budget for a single window.
However, this condition applies equally to all methods, ensuring a fair
comparison.

We observe that MTSP-LDP consistently achieves the lowest estimation errors
across all settings, with particularly pronounced improvements in MRE,
demonstrating superior accuracy under normalized error metrics.
This is because other approaches do not account for the complexity of performing
complex query tasks on streaming data and instead rely solely on publishing
frequency histograms.
While frequency histograms can be used to derive results for complex queries,
their cumulative error grows proportionally with the query range, leading to a
sharp decline in data utility.
In contrast, MTSP-LDP's data-adaptive hierarchical structure and cross-timestamp
processing enable more precise estimation for such queries, especially under
tight privacy budgets.

\subsubsection{Evaluation for Event Monitoring}

For event monitoring, we define a specific task where the goal is to detect the
change in the range query in the entire value domain $\Omega$ with threshold
$\vartheta$ equal to the median of $\Omega$ in each dataset.
The two functions in the range query are defined as follows.
\begin{align*}
  \alpha(\Omega,\Delta) &= Q_t^r(\Omega,\Delta)-Q_{t-w+1}^r(\Omega,\Delta)\\
  \beta(x,\vartheta) &= \indr{x>\vartheta}
\end{align*}
By analyzing the changes across consecutive windows, we ensure a consistent
comparison of each algorithm's ability to capture dynamic variations in the data
streams.
To better understand each method's performance, we vary the threshold
$\vartheta$, and obtain the ROC curve for each dataset, as illustrated in
\cref{fig:ROC_query}.

\begin{figure*}[t]
  \centering
  \includegraphics[width=.4\linewidth]{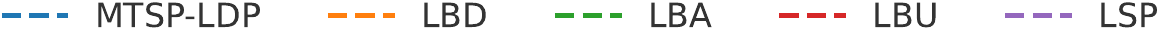}\\[-2ex]
  \subfloat[Cosmetics]{%
    \includegraphics[width=.22\linewidth]{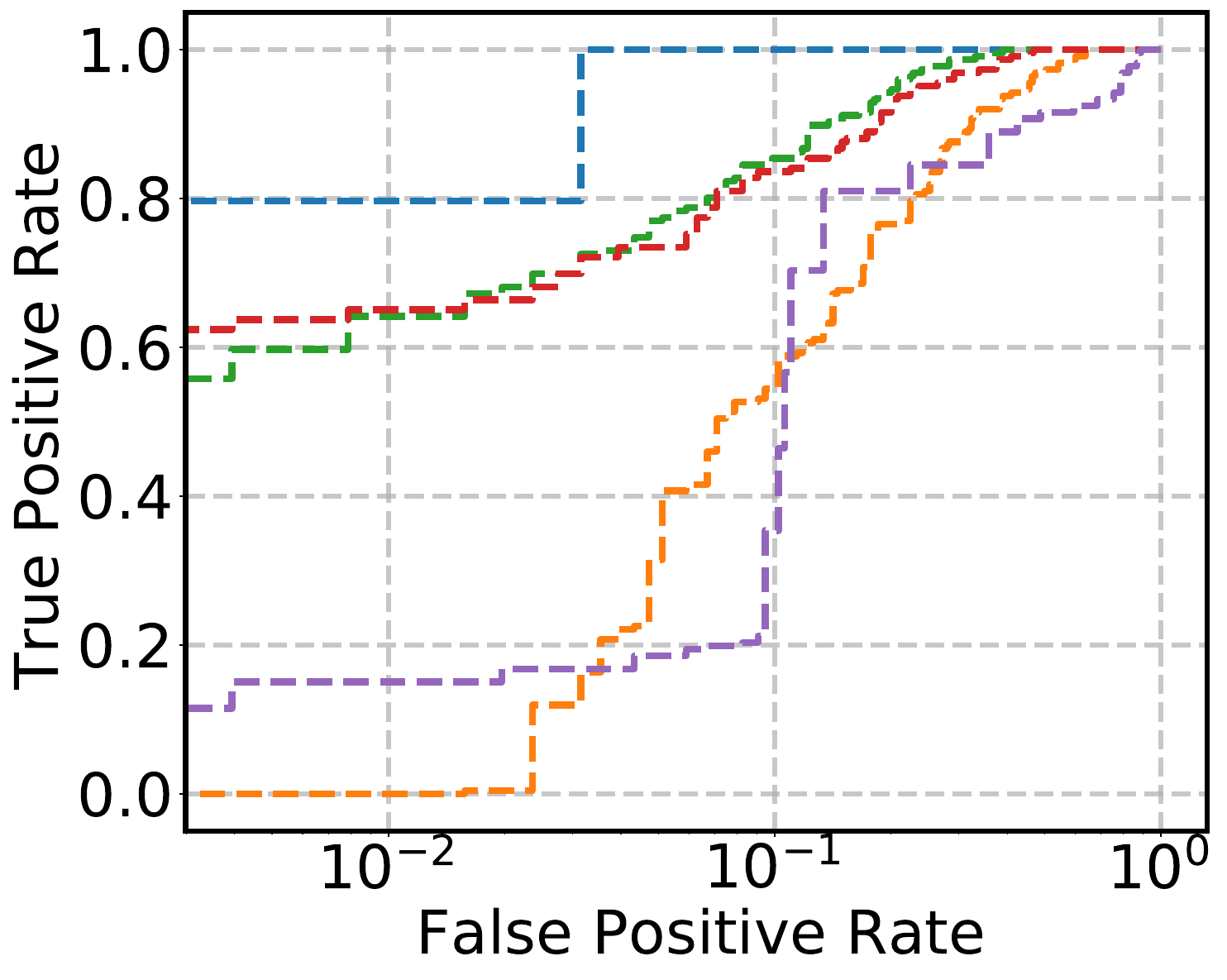}}\quad
  \subfloat[Taxi]{%
    \includegraphics[width=.22\linewidth]{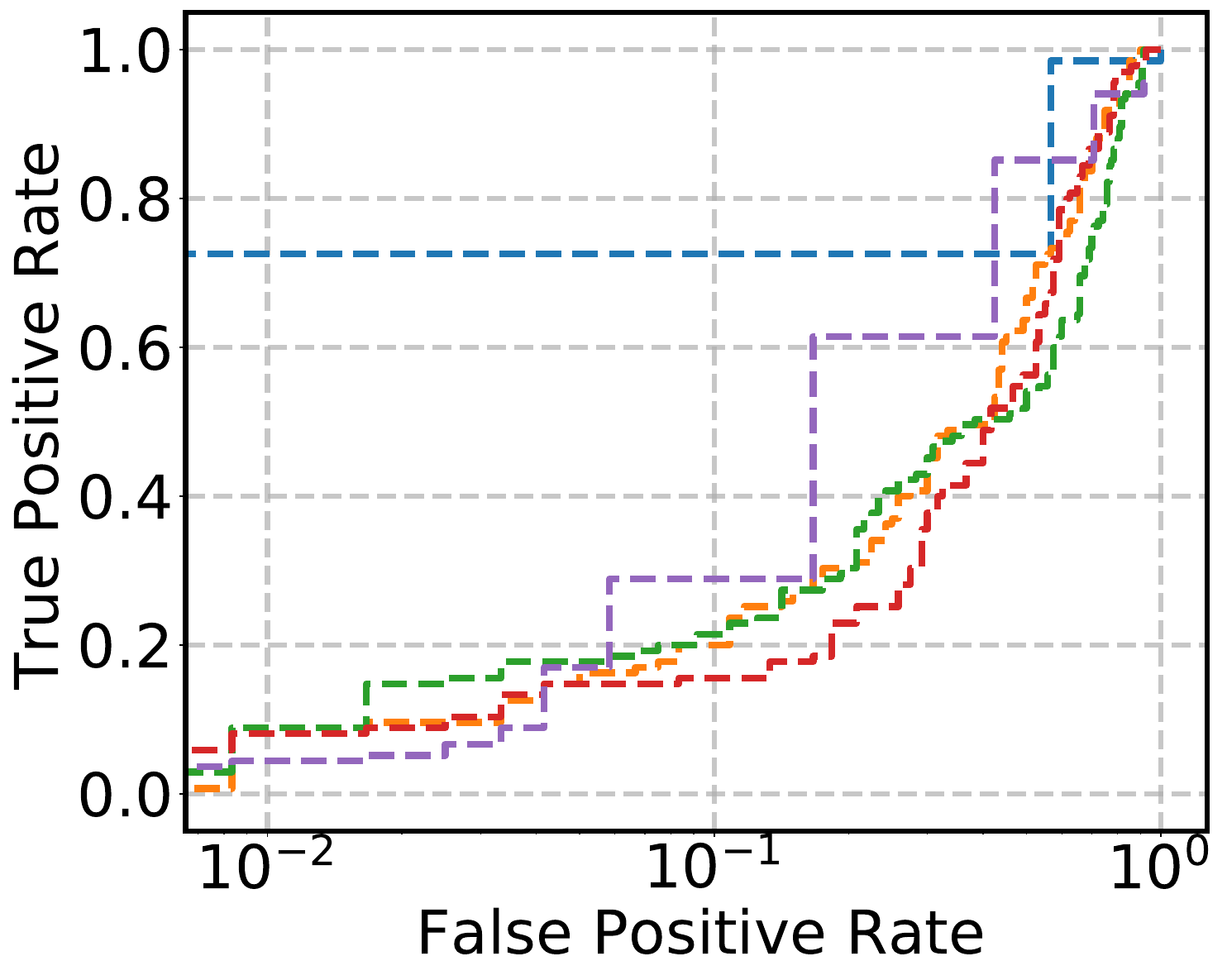}}\quad
  \subfloat[Loan]{%
    \includegraphics[width=.22\linewidth]{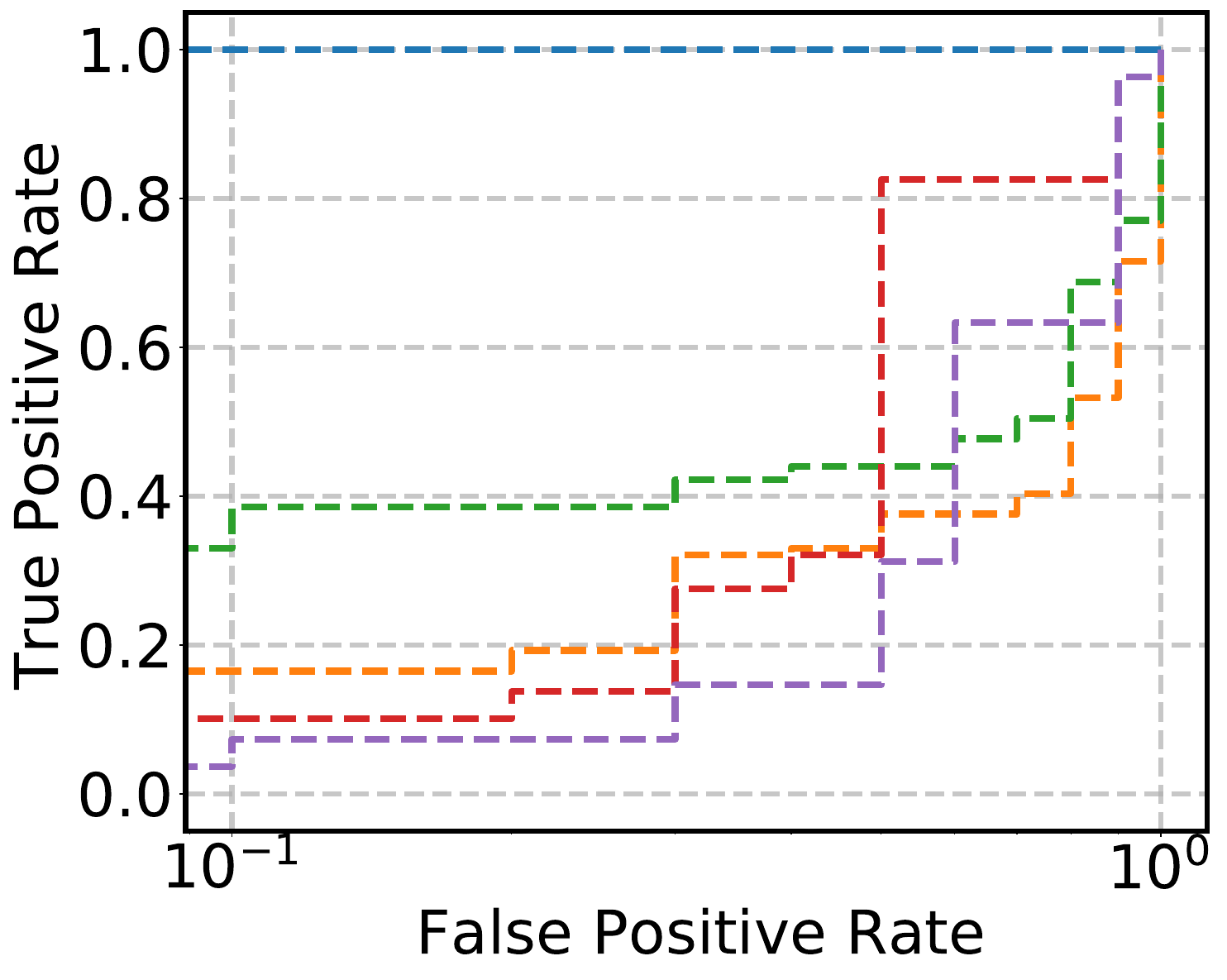}}\quad
  \subfloat[Foursquare]{%
    \includegraphics[width=.22\linewidth]{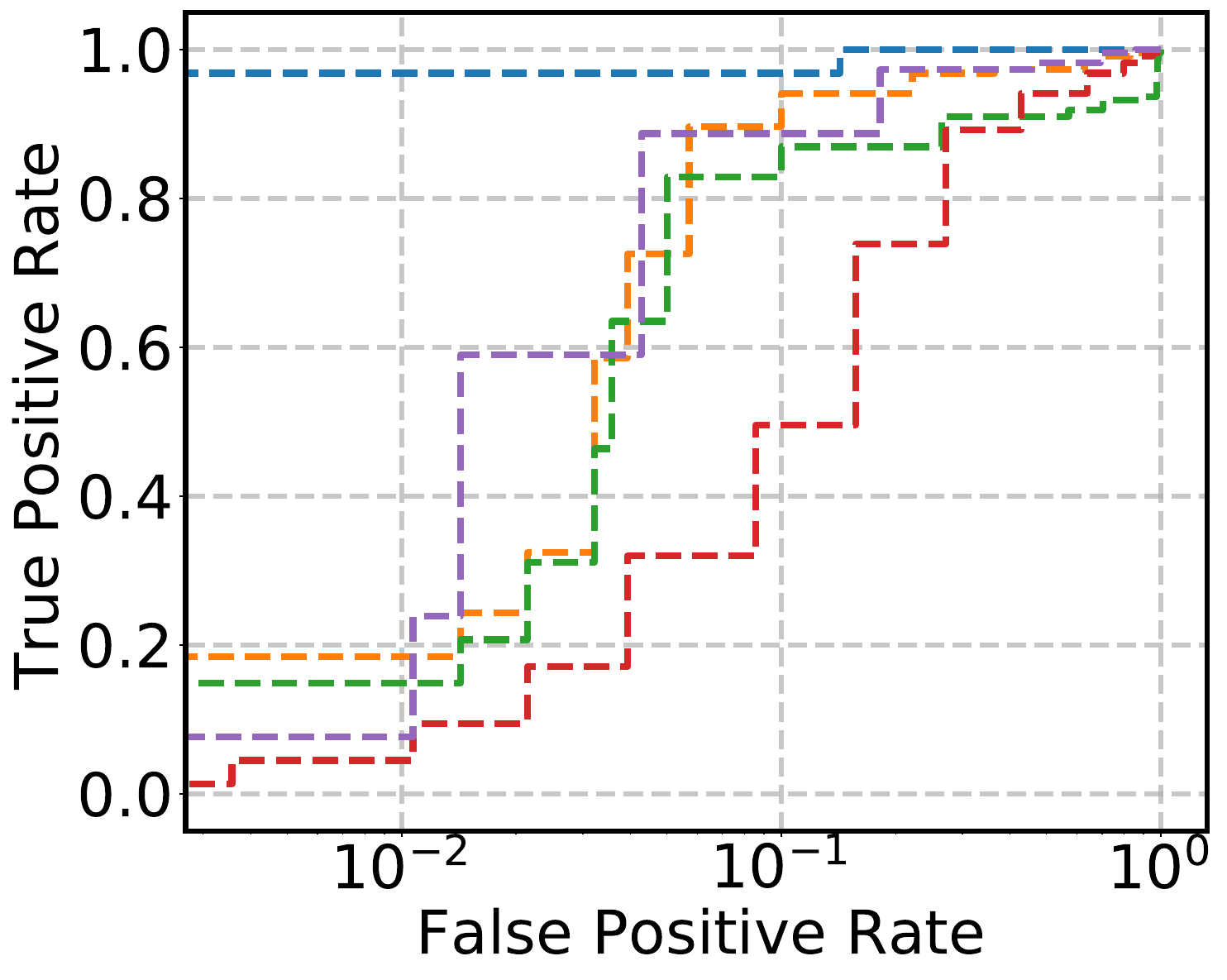}}
  \caption{ROC for event monitoring with $\epsilon=1$ and $w=20$}
  \label{fig:ROC_query}
  \Description{The figure displays four Receiver Operating Characteristic (ROC) curves corresponding to the Cosmetics, Taxi, Loan, and Foursquare datasets. 
The X-axis represents the False Positive Rate on a logarithmic scale, and the Y-axis represents the True Positive Rate. 
Five methods are compared: the proposed MTSP-LDP (blue dashed line) and four baselines (LBD, LBA, LBU, LSP). 
In all four subplots, the curve for MTSP-LDP rises most steeply and stays closest to the top-left corner, indicating it achieves the highest True Positive Rate while maintaining a low False Positive Rate. 
The baseline methods generally lag significantly behind, requiring much higher False Positive Rates to achieve comparable detection performance.}
\end{figure*}

We observe that MTSP-LDP consistently outperforms the other methods, attaining
near-perfect TPR across the full range of FPR.
This demonstrates the robustness and adaptability of MTSP-LDP to different data
distributions.
At low FPR levels (FPR$<0.1$), MTSP-LDP achieves significantly higher TPR
compared to other methods, highlighting its superior sensitivity under strict
false positive constraints.
In contrast, LBD and LBA exhibit limited performance, with slower TPR growth as
FPR increases.
While LSP shows relatively strong performance on the cosmetics dataset,
approaching MTSP-LDP in some regions, it still falls short of MTSP-LDP's
consistent accuracy across all datasets.

These results underscore the effectiveness of MTSP-LDP in addressing complex
query tasks under varying privacy constraints and data characteristics.
Its ability to achieve high accuracy across diverse datasets highlights its
suitability for real-world applications requiring stringent privacy guarantees
and robust performance.

%% file: samples/sigmod2026/exp_counting_query_figures.tex
\begin{figure*}[t]
  \centering
  \includegraphics[width=.4\linewidth]{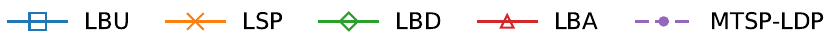}

  \subfloat[Cosmetics ($w=20$)]{%
    \includegraphics[width=.25\linewidth]{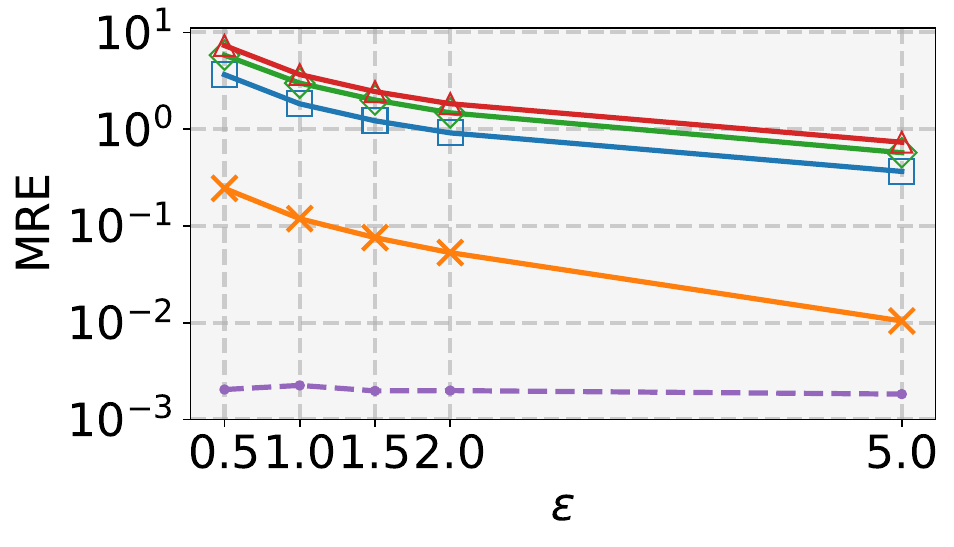}}
  \subfloat[Cosmetics ($\epsilon=1$)]{%
    \includegraphics[width=.25\linewidth]{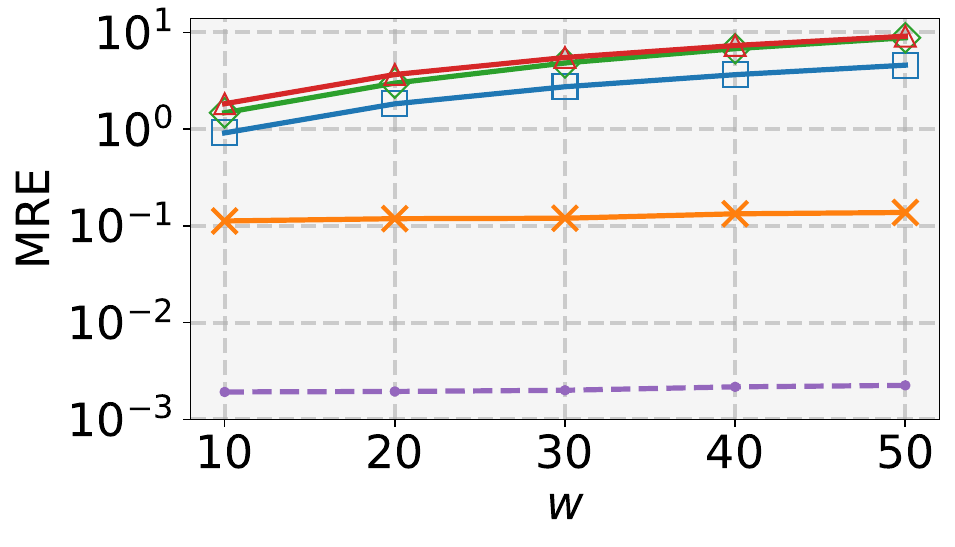}}
  \subfloat[Cosmetics ($w=20$)]{%
    \includegraphics[width=.25\linewidth]{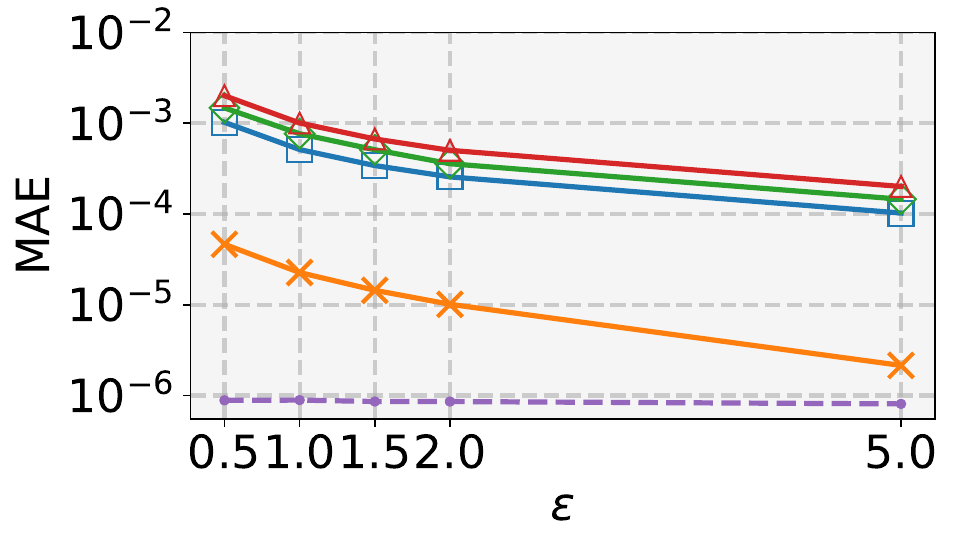}}
  \subfloat[Cosmetics ($\epsilon=1$)]{%
    \includegraphics[width=.25\linewidth]{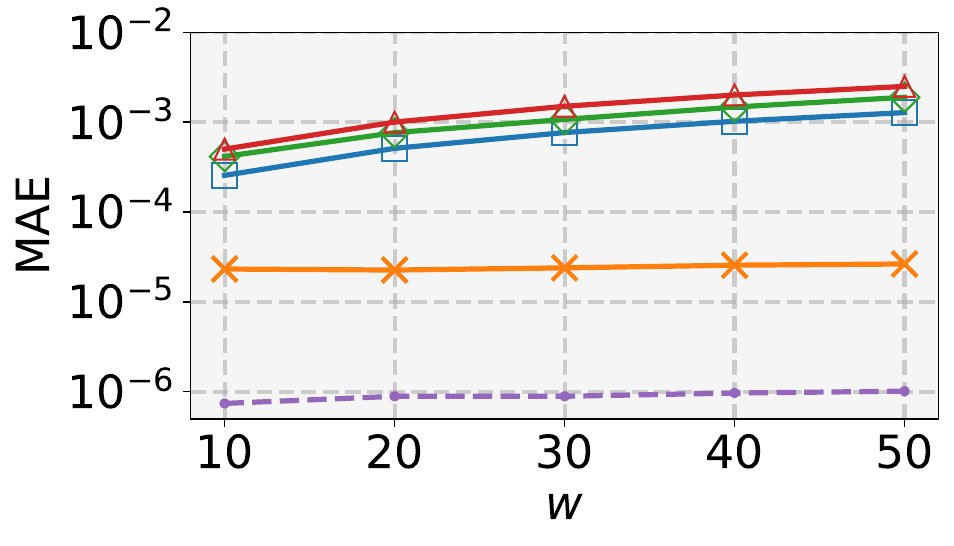}}

  \subfloat[Taxi ($w=20$)]{%
    \includegraphics[width=.25\linewidth]{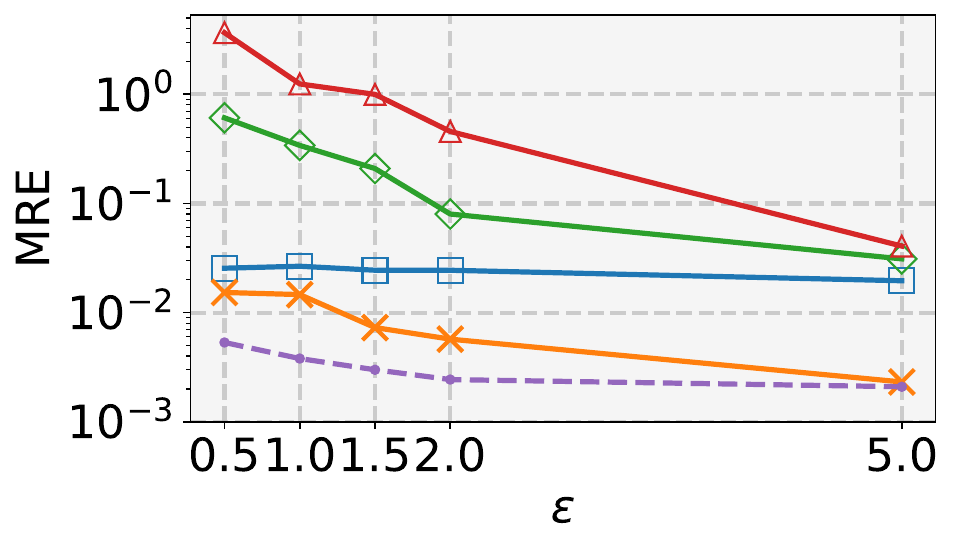}}
  \subfloat[Taxi ($\epsilon=1$)]{%
    \includegraphics[width=.25\linewidth]{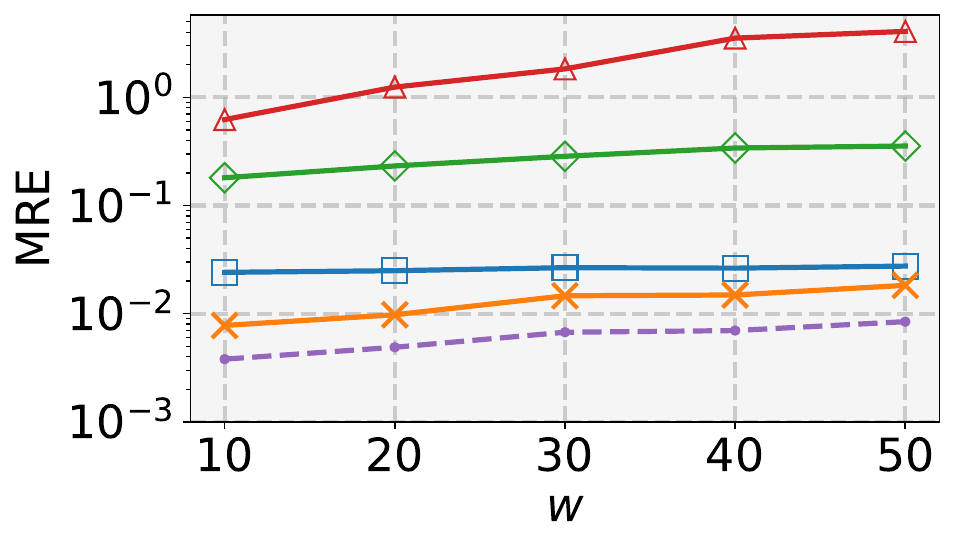}}
  \subfloat[Taxi ($w=20$)]{%
    \includegraphics[width=.25\linewidth]{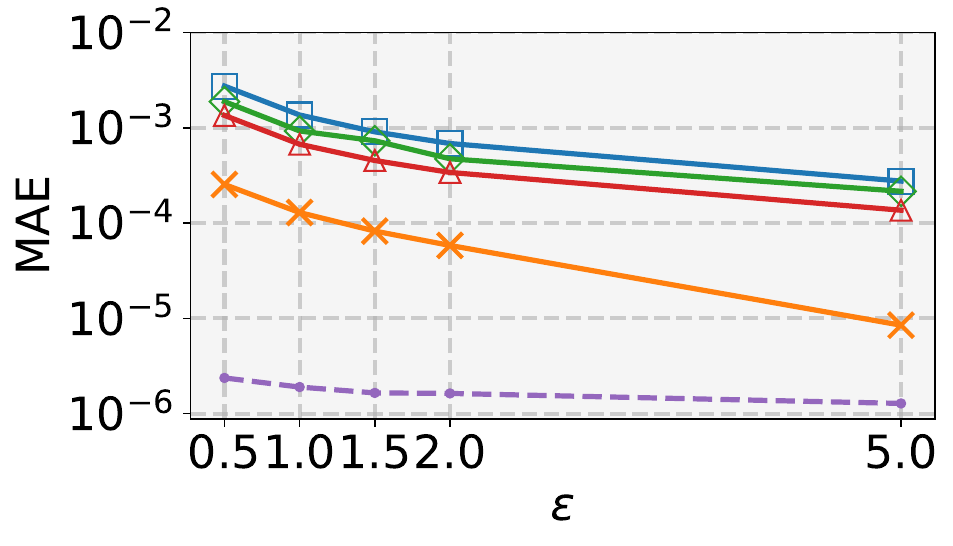}}
  \subfloat[Taxi ($\epsilon=1$)]{%
    \includegraphics[width=.25\linewidth]{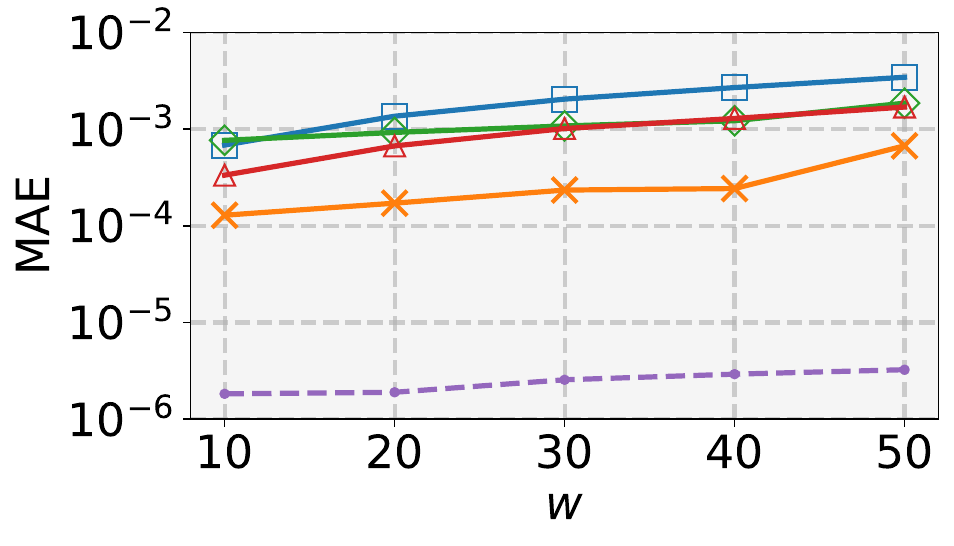}}

  \subfloat[Loan ($w=20$)]{%
    \includegraphics[width=.25\linewidth]{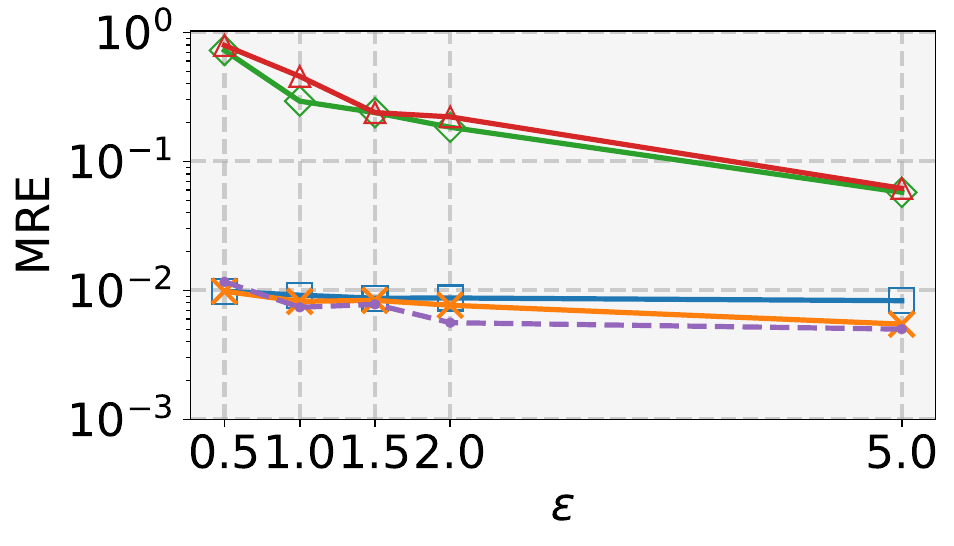}}
  \subfloat[Loan ($\epsilon=1$)]{%
    \includegraphics[width=.25\linewidth]{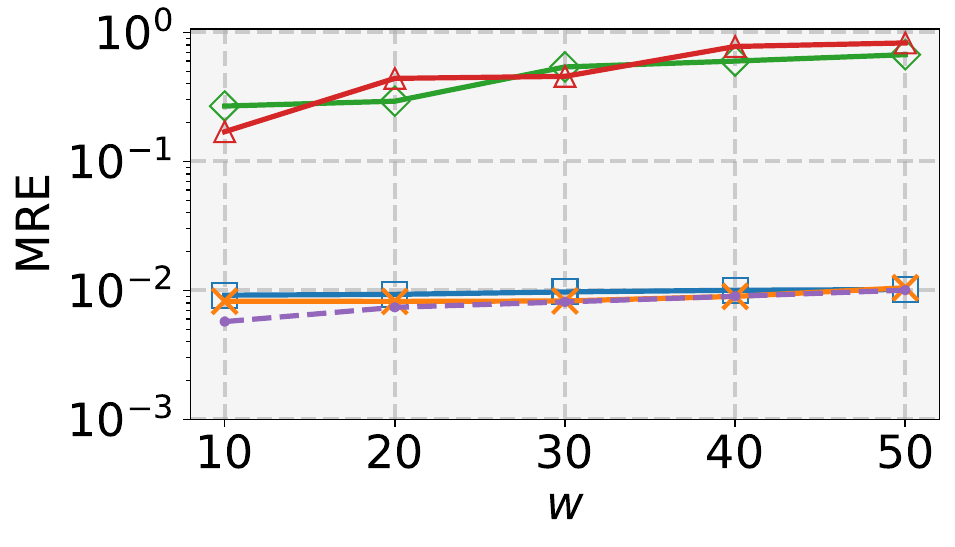}}
  \subfloat[Loan ($w=20$)]{%
    \includegraphics[width=.25\linewidth]{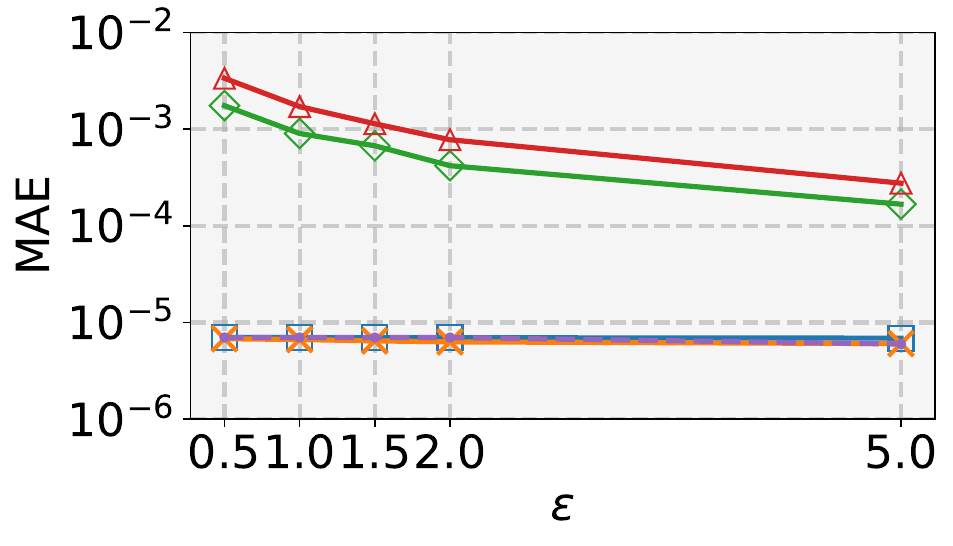}}
  \subfloat[Loan ($\epsilon=1$)]{%
    \includegraphics[width=.25\linewidth]{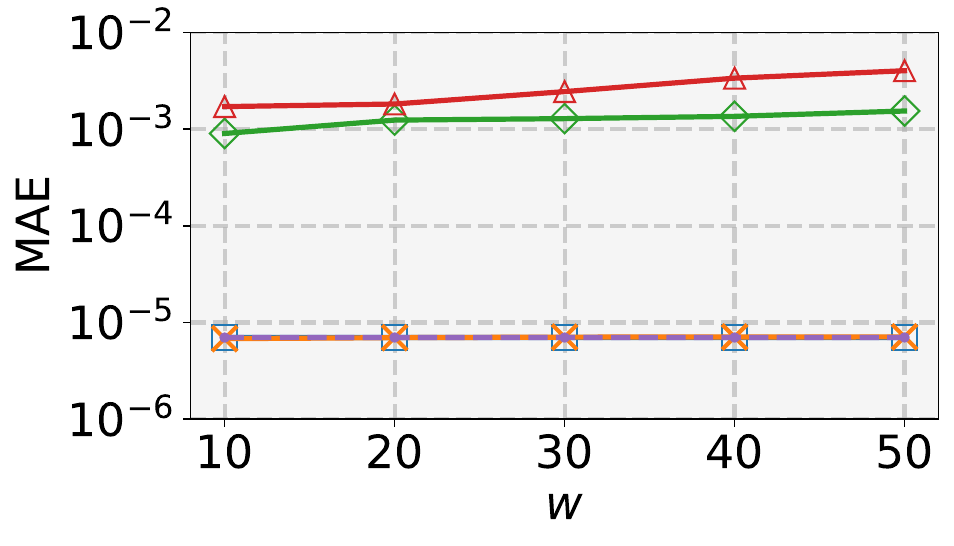}}

  \subfloat[Foursquare ($w=20$)]{%
    \includegraphics[width=.25\linewidth]{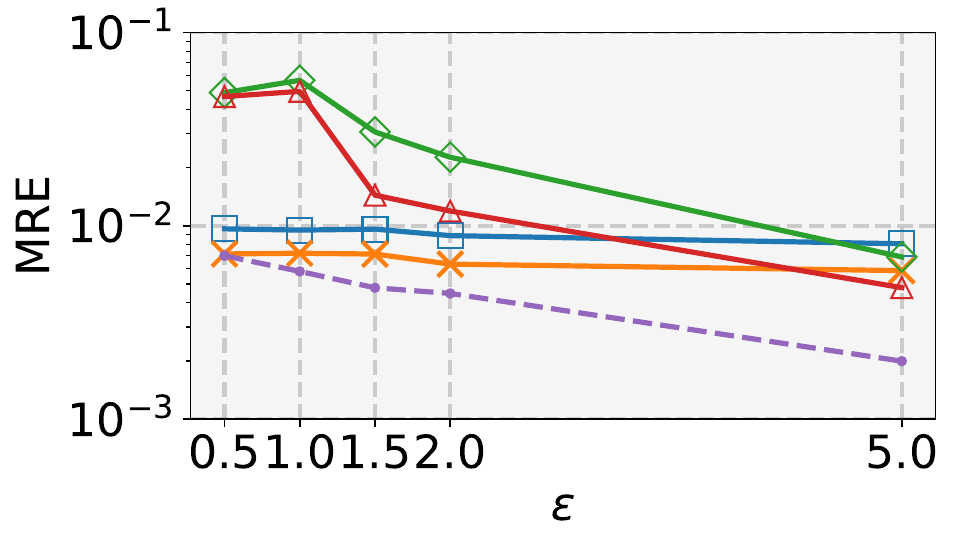}}
  \subfloat[Foursquare ($\epsilon=1$)]{%
    \includegraphics[width=.25\linewidth]{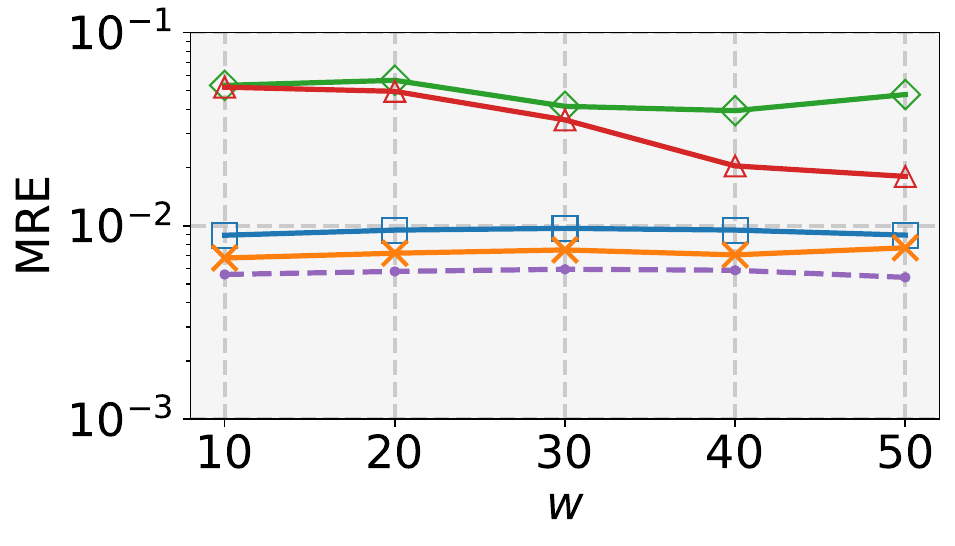}}
  \subfloat[Foursquare ($w=20$)]{%
    \includegraphics[width=.25\linewidth]{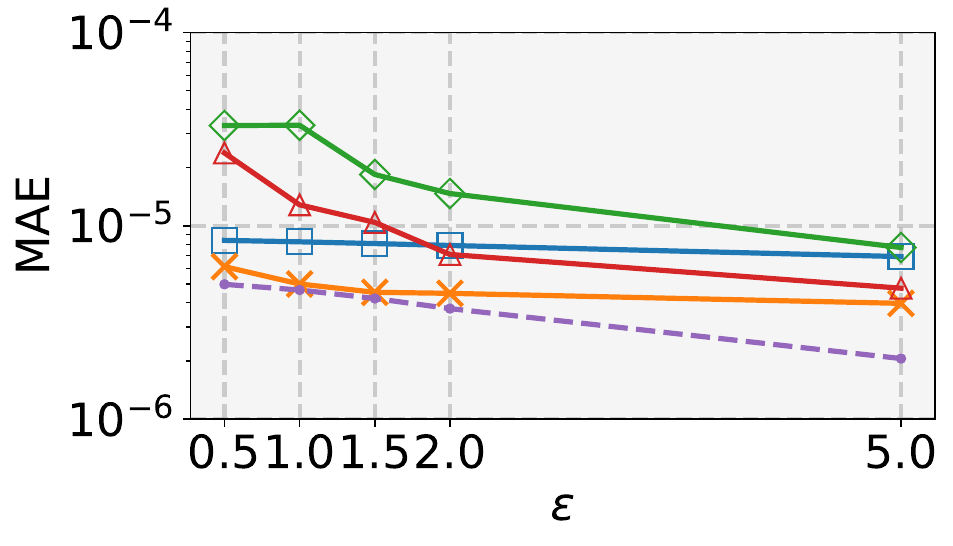}}
  \subfloat[Foursquare ($\epsilon=1$)]{%
    \includegraphics[width=.25\linewidth]{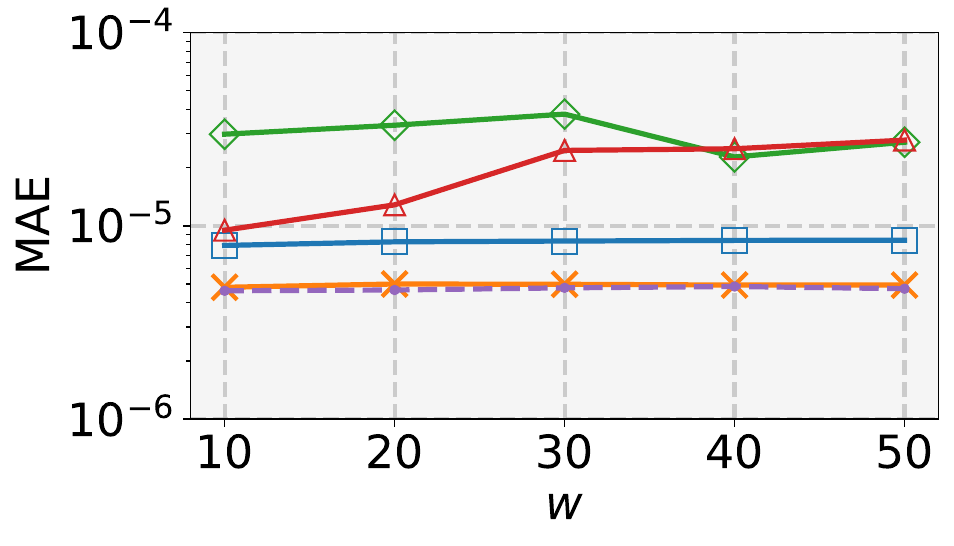}}
  \caption{Counting query evaluation with varying $\epsilon$ and $w$}
  \label{fig:histogram}
  \Description{The figure presents a 4x4 grid of line charts evaluating the performance of counting queries. 
The rows correspond to four datasets: Cosmetics, Taxi, Loan, and Foursquare. 
The columns represent different metrics and varying parameters: the first and second columns show Mean Relative Error (MRE) varying epsilon and window size w, respectively; the third and fourth columns show Mean Absolute Error (MAE) varying epsilon and w. 
Five methods are compared: LBU, LSP, LBD, LBA, and the proposed MTSP-LDP. 
In all 16 subplots, the Y-axis is on a logarithmic scale. 
Visually, the proposed MTSP-LDP method (represented by a purple dashed line) consistently stays at the bottom of all charts, indicating it achieves the lowest error by a significant margin (often orders of magnitude) compared to the four baseline methods across all datasets and parameter settings.}
\end{figure*}

%% file: samples/sigmod2026/exp_range_query_figures.tex
\begin{figure*}[t]
  \centering
  \includegraphics[width=.4\linewidth]{figs/legend_methods.pdf}

  \subfloat[Cosmetics ($w=20$)]{%
    \includegraphics[width=.25\linewidth]{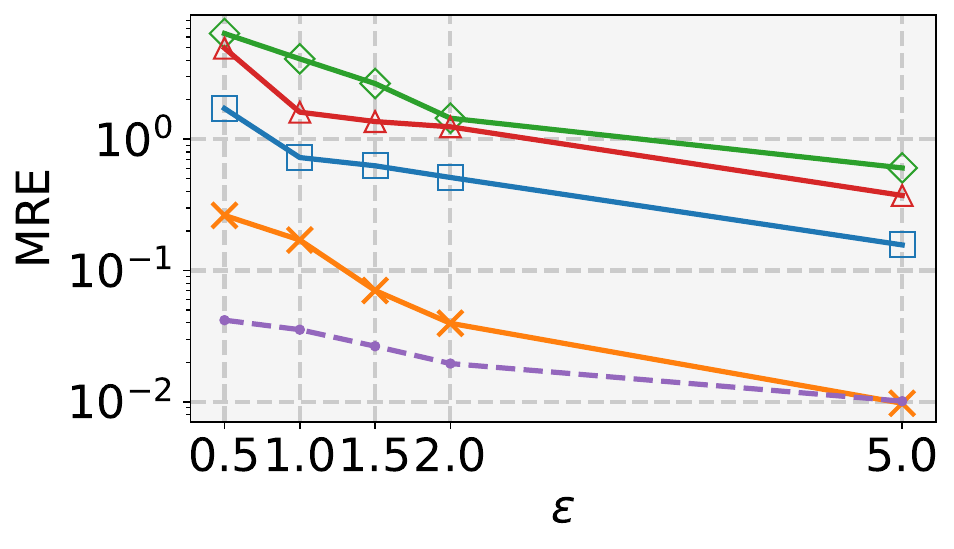}}
  \subfloat[Cosmetics ($\epsilon=1$)]{%
    \includegraphics[width=.25\linewidth]{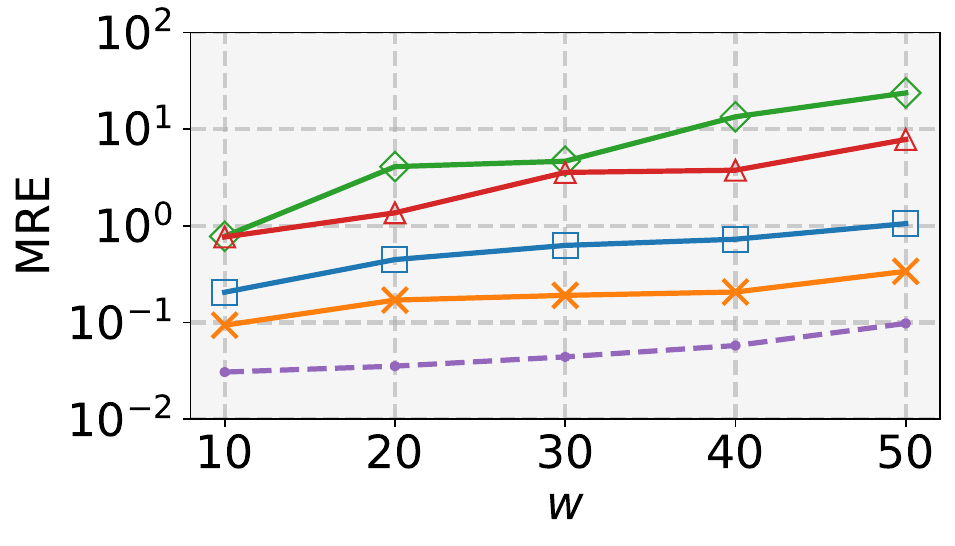}}
  \subfloat[Cosmetics ($w=20$)]{%
    \includegraphics[width=.25\linewidth]{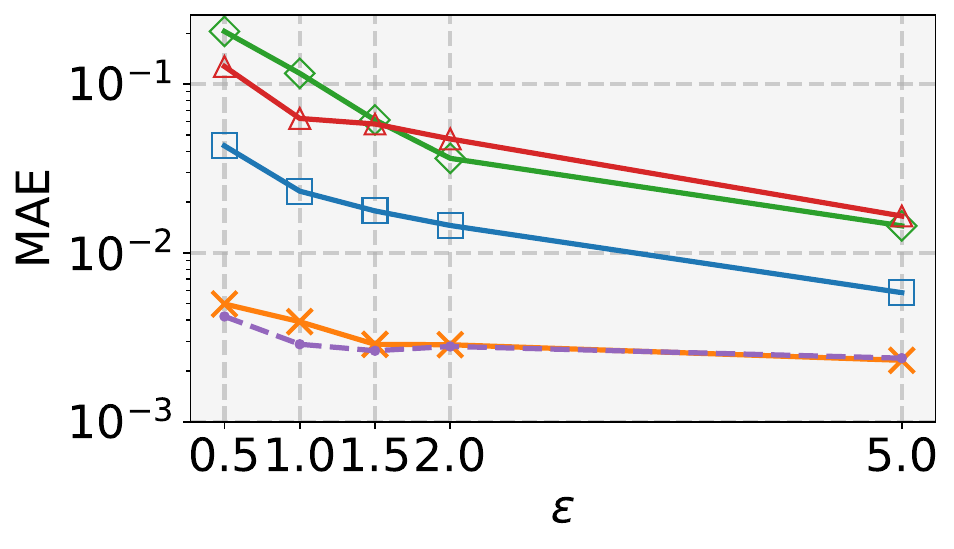}}
  \subfloat[Cosmetics ($\epsilon=1$)]{%
    \includegraphics[width=.25\linewidth]{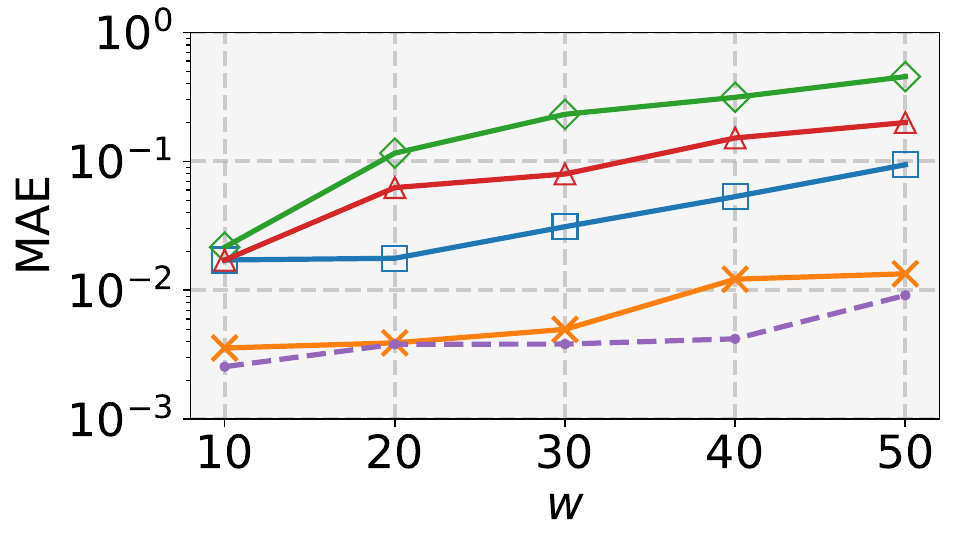}}

  \subfloat[Taxi ($w=20$)]{%
    \includegraphics[width=.25\linewidth]{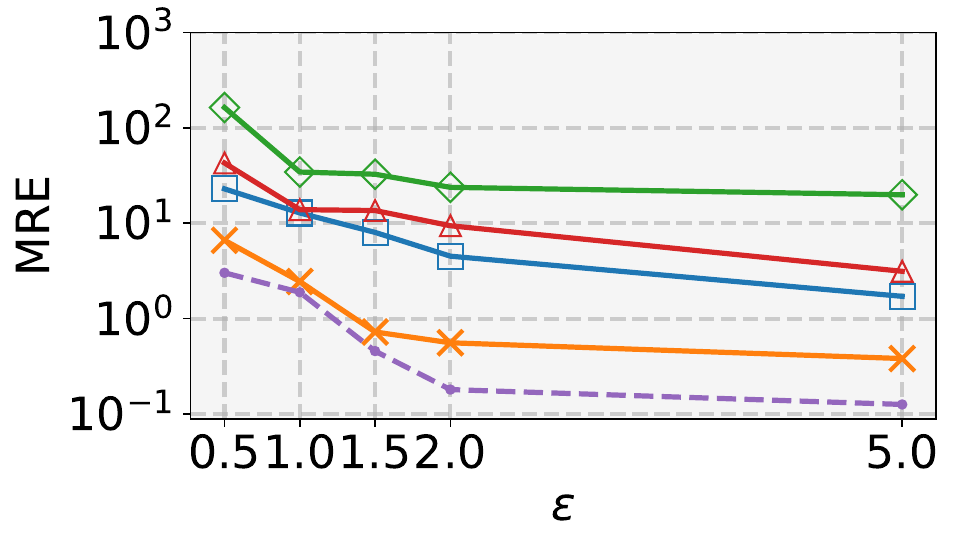}}
  \subfloat[Taxi ($\epsilon=1$)]{%
    \includegraphics[width=.25\linewidth]{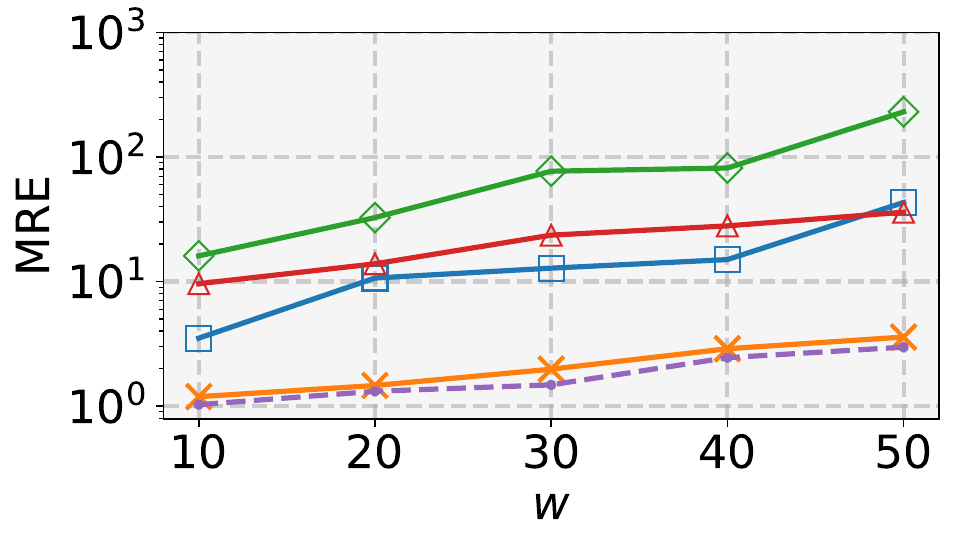}}
  \subfloat[Taxi ($w=20$)]{%
    \includegraphics[width=.25\linewidth]{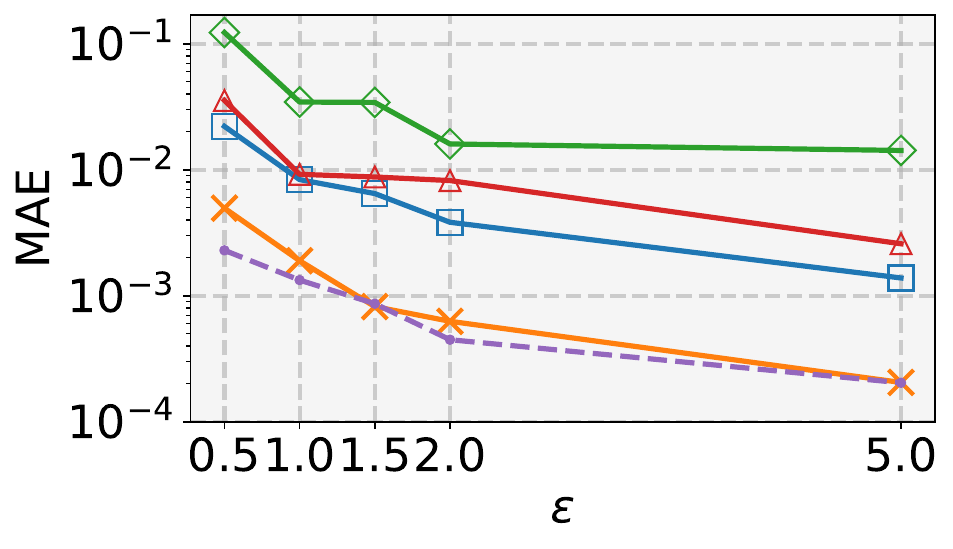}}
  \subfloat[Taxi ($\epsilon=1$)]{%
    \includegraphics[width=.25\linewidth]{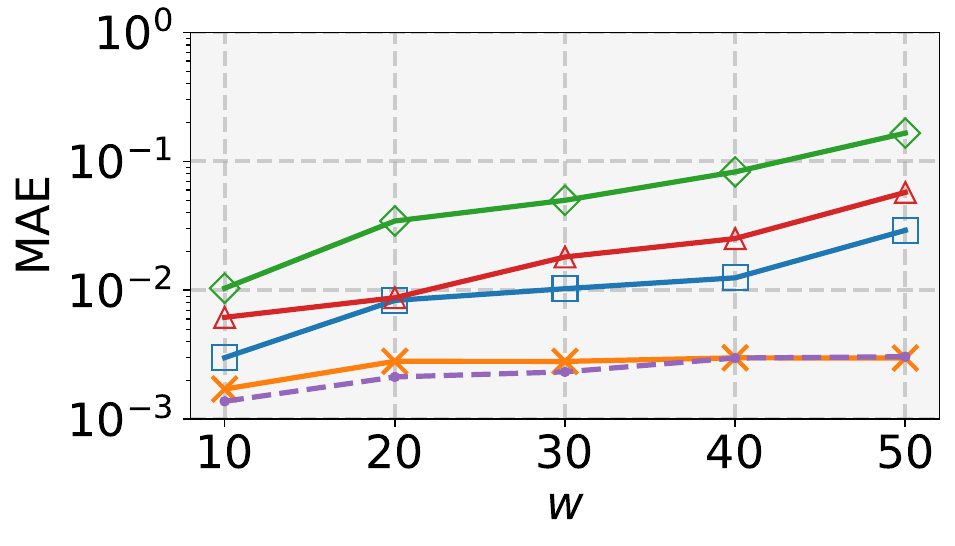}}

  \subfloat[Loan ($w=20$)]{%
    \includegraphics[width=.25\linewidth]{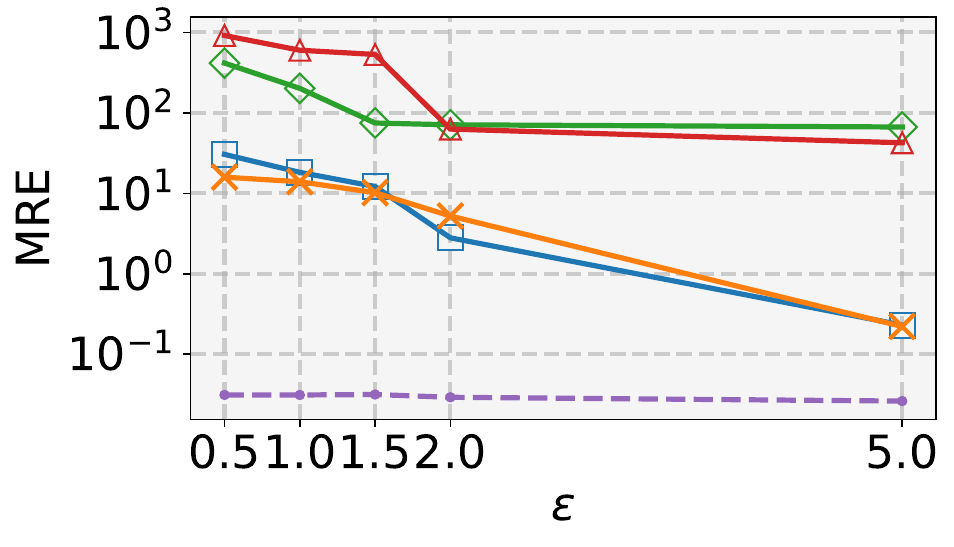}}
  \subfloat[Loan ($\epsilon=1$)]{%
    \includegraphics[width=.25\linewidth]{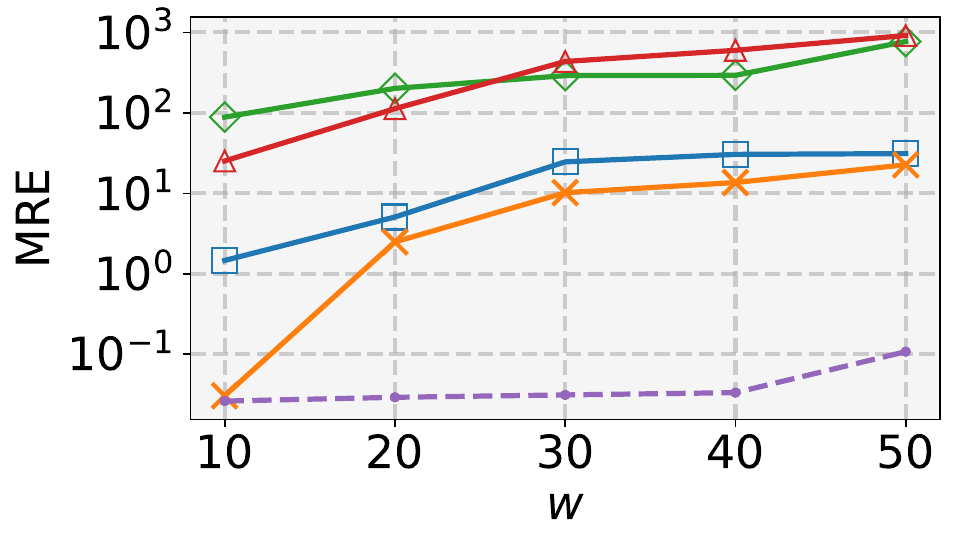}}
  \subfloat[Loan ($w=20$)]{%
    \includegraphics[width=.25\linewidth]{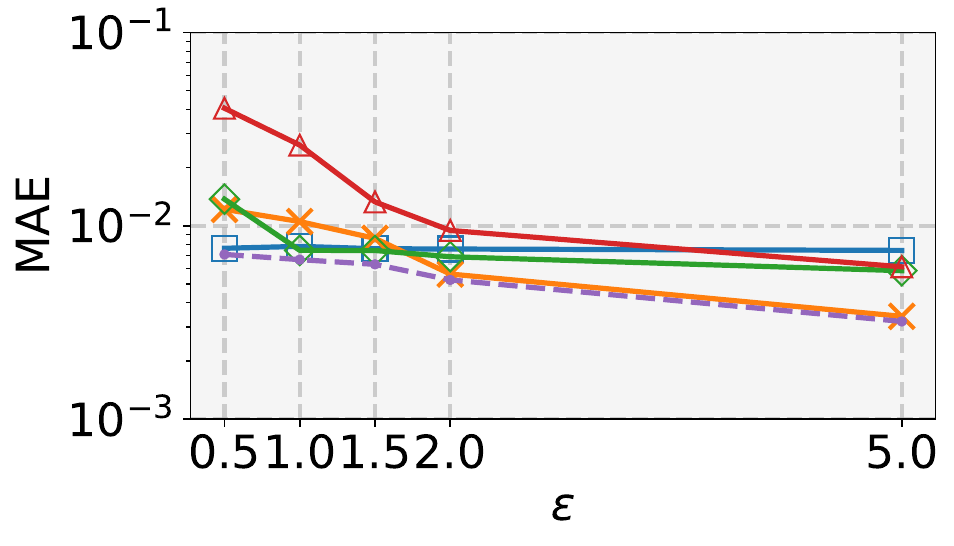}}
  \subfloat[Loan ($\epsilon=1$)]{%
    \includegraphics[width=.25\linewidth]{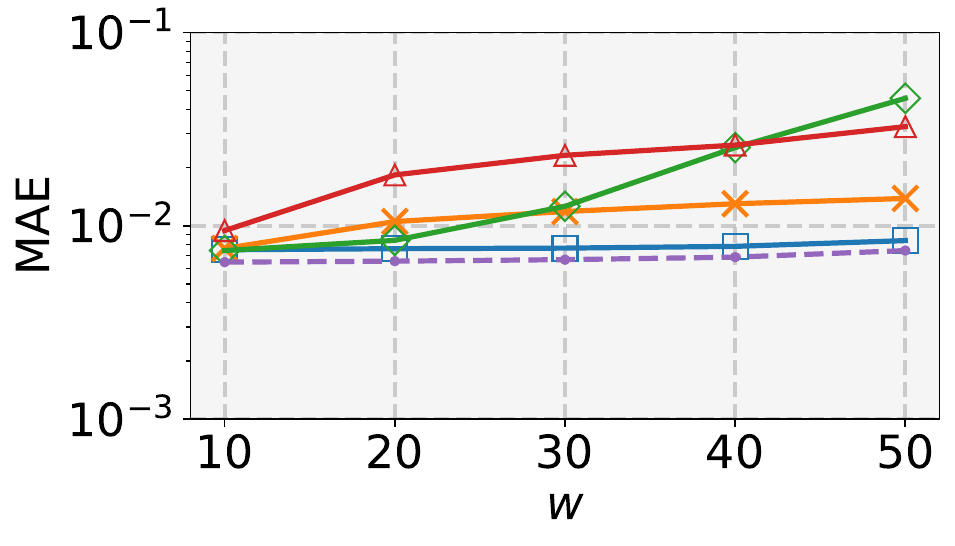}}

  \subfloat[Foursquare ($w=20$)]{%
    \includegraphics[width=.25\linewidth]{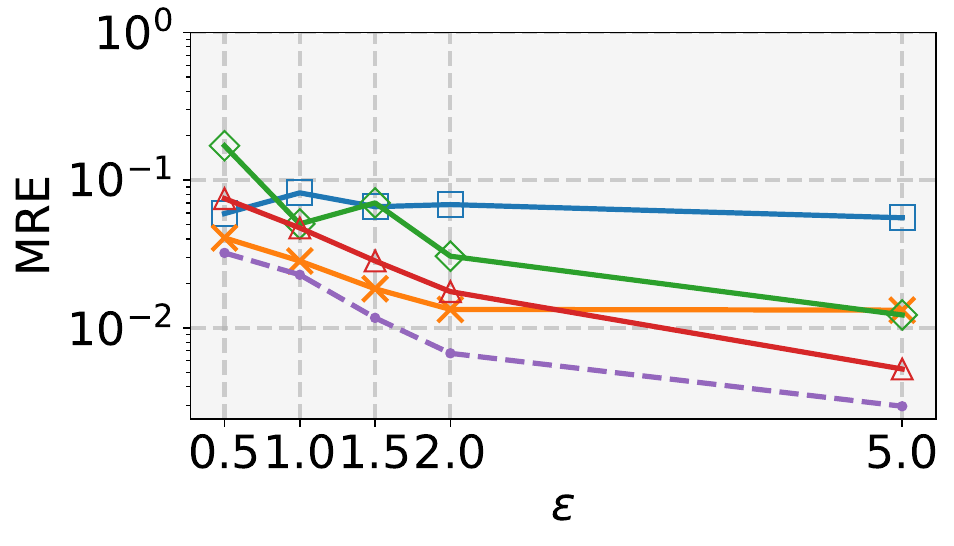}}
  \subfloat[Foursquare ($\epsilon=1$)]{%
    \includegraphics[width=.25\linewidth]{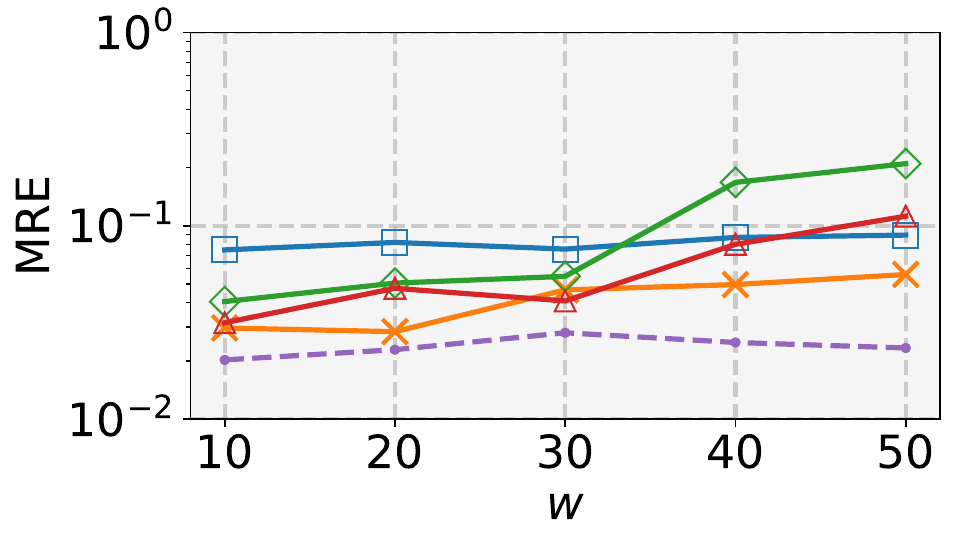}}
  \subfloat[Foursquare ($w=20$)]{%
    \includegraphics[width=.25\linewidth]{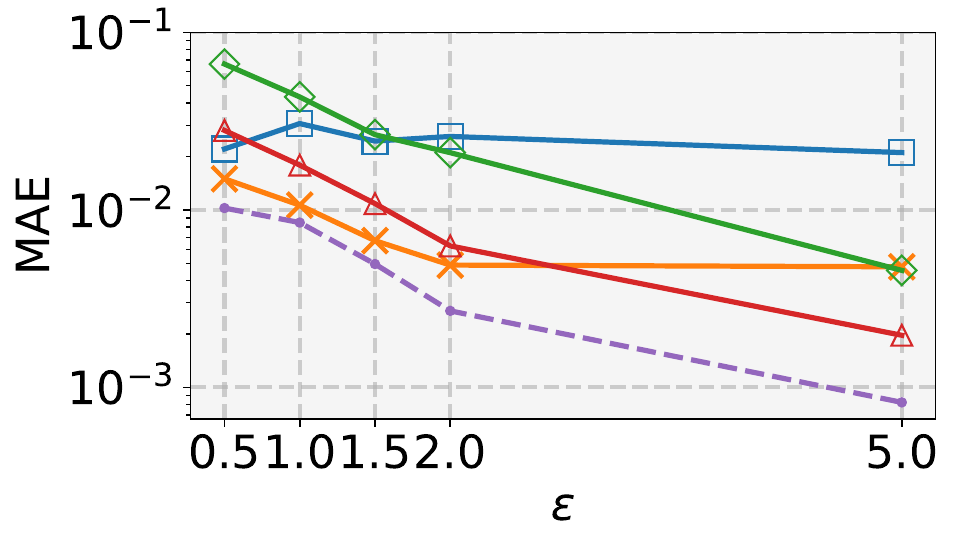}}
  \subfloat[Foursquare ($\epsilon=1$)]{%
    \includegraphics[width=.25\linewidth]{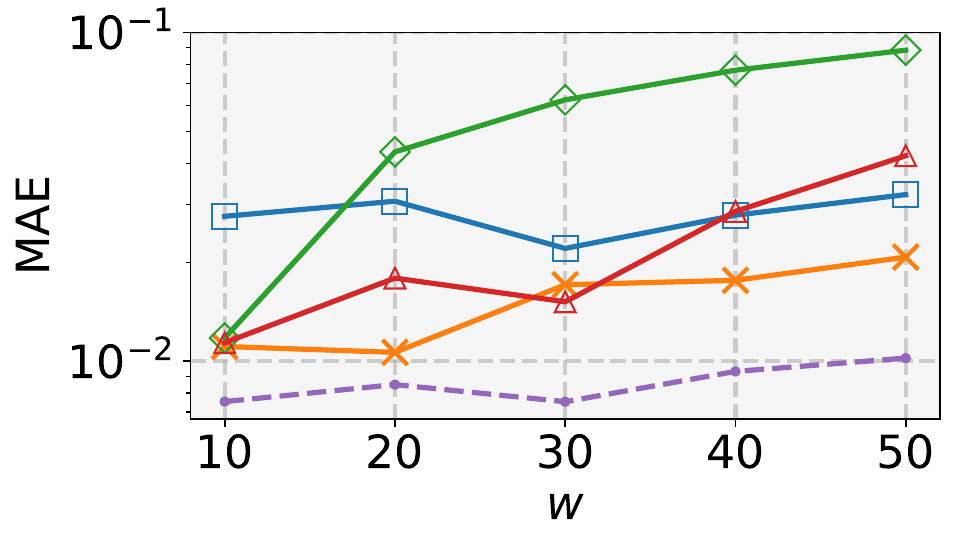}}
  \caption{Range query evaluation with varying $\epsilon$ and $w$}
  \label{fig:range_query}
  \Description{The figure presents a 4x4 grid of line charts evaluating the performance of range query. 
The rows correspond to four datasets: Cosmetics, Taxi, Loan, and Foursquare. 
The columns represent different metrics and varying parameters: the first and second columns show Mean Relative Error (MRE) varying epsilon and window size w, respectively; the third and fourth columns show Mean Absolute Error (MAE) varying epsilon and w. 
Five methods are compared: LBU, LSP, LBD, LBA, and the proposed MTSP-LDP. 
In all 16 subplots, the Y-axis is on a logarithmic scale. 
Visually, the proposed MTSP-LDP method (represented by a purple dashed line) consistently stays at the bottom of all charts, indicating it achieves the lowest error by a significant margin (often orders of magnitude) compared to the four baseline methods across all datasets and parameter settings.}
\end{figure*}

%% file: samples/sigmod2026/related.tex
\section{Related Work}
\label{sec:related}

We summarize some related work in the literature.

\paragraph{CDP Methods}
In the centralized setting, Dwork et al.~introduced event-level DP for
continuously releasing statistics like counts and histograms~\cite{Dwork2006}.
Hierarchical tree structures are widely used for temporal range queries under CDP. In such trees, leaf nodes store the noisy counts for each timestamp, while internal nodes store noisy sums over the intervals they cover.
For fixed-length binary streams, a binary tree structure was used to aggregate
counts~\cite{Dwork2010}, and later extended to infinite streams by Chan et
al.~through consistency constraints~\cite{Chan2011}.
To further reduce noise in sparse regions, Dwork proposed adaptive partitioning
based on thresholds~\cite{Dwork2015}.
However, this method is limited to post-partition release and lacks real-time
applicability.
Chen et al.~\cite{Chen2017} improved flexibility by modifying monitored events.
Cao et al.~\cite{Cao2015} introduced group-based histogram publishing, adding
Laplace noise to similar time slots.
Although these techniques support numerical data, they require raw data access
and thus cannot be directly applied under LDP.
Bao et al.~\cite{Bolot2013} addressed sliding-window predicate sum queries using
decay models.

For user-level DP, Fan et al.~\cite{Fan2013a} proposed FAST, a real-time
aggregation framework using adaptive sampling and filtering.
It was extended to 2D monitoring via spatial partitioning with
quadtrees~\cite{Fan2013}.
Rastogi et al.~\cite{Rastogi2010} developed a method based on discrete Fourier
transform (DFT), though it is suited only for offline analysis. {The analysis in \cite{2023userdp} reveals that on infinite streams, achieving user-level DP forces error to grow without bound as the maximum per-user prefix contribution increases over time. While they control privacy by truncating per-user contributions, this approach in turn introduces bias.

Kellaris et al.~\cite{Kellaris2014} introduced the $w$-event DP model to protect
sequences of events over sliding windows and proposed BA/BD mechanisms.
However, they rely on raw data to compute similarity, exposing them to inference
attacks.
Cao et al.~\cite{Cao2015} extended BD to support variable-length trajectories,
using a greedy strategy to match current data with historical outputs.
This resulted in uneven privacy budget distribution.
Du et al.~\cite{du2025} extended BD/BA to the personalized $w$-event privacy setting, where different users have heterogeneous privacy requirements, and proposed PBD and PBA.
Wang et al.~\cite{Wang2019} proposed E-RescueDP, which adaptively allocates
privacy budgets using an RNN to support real-time release.
Li et al.~\cite{Li2025} proposed SPAS, which predicts future data variation to
adaptively determine data sampling and privacy budget allocation under $w$-event
DP.

\paragraph{LDP Methods}
Most existing LDP works focus on single-value counting and frequency estimation
of static data, with only a few studies focusing on stream data analysis.
Memoization-based techniques~\cite{Ding2017, Erlingsson2014, Arcolezi2021} were
proposed to offer longitudinal LDP guarantees.
Joseph et al.~\cite{Joseph2018} estimated stream means by having users compare
their local averages against the last released value and vote on updates.
Although satisfying event-level LDP, this method assumes Bernoulli input and
temporal independence, limiting its generality.
THRESH further assumes that the number of global updates is bounded by
distribution changes, making it unsuitable for infinite streams.
Wang et al.~\cite{Wang2021} proposed an event-level LDP framework for interval
sum estimation using a hybrid mechanism with thresholding.
However, it directly extends CDP-style hierarchies, offering limited protection
for unbounded streams.
Bao et al.~\cite{Bao2021} leveraged autocorrelation to reduce noise via an
analytic Gaussian mechanism, but their method applies only to finite data and
achieves approximate $(\epsilon, \delta)$-LDP under periodic budget renewal.
Erlingsson et al.~\cite{Erlingsson2019} introduced a shuffling model for
correlated time series under user-level LDP, but their approach assumes
integer-valued inputs with bounded updates, and the hierarchical design
restricts scalability to infinite streams.

%% file: samples/sigmod2026/conclusion.tex
\section{Conclusion} \label{sec:conclusions}

We propose MTSP-LDP, a $w$-event LDP framework to handle infinite data streams
and support multiple streaming query tasks, including counting queries, range
queries, and event monitoring.
MTSP-LDP leverages a novel OBA algorithm that can dynamically allocate privacy
budgets within a window.
MTSP-LDP then constructs a private data-adaptive tree to support complex queries
more accurately.
Experiments on real-world datasets demonstrate that MTSP-LDP significantly
outperforms state-of-the-art methods.
Future work could focus on further enhancing MTSP-LDP, such as extending the
framework to handle multidimensional data streams and optimizing its
adaptability for even broader application scenarios.